\documentclass[sigconf,nonacm]{acmart}

\renewcommand\footnotetextcopyrightpermission[1]{}

\usepackage{amsmath,amsfonts,amsthm}
\usepackage{algorithmic}
\usepackage{graphicx}
\usepackage{textcomp}
\usepackage{xcolor}
\usepackage{subcaption}
\usepackage{multirow}
\usepackage{enumitem}
\usepackage{units}
\usepackage[ruled,linesnumbered]{algorithm2e}
\usepackage{balance}
\usepackage{tabularx}
\usepackage{adjustbox}
\usepackage{empheq,eqparbox}
\usepackage{makecell}
\usepackage{array}
\usepackage{xurl}
\usepackage{hyperref}
\usepackage{framed}

\theoremstyle{remark}
\newtheorem{example}{Example}
\theoremstyle{plain}
\newtheorem{assumption}{Assumption}
\newtheorem{lemma}{Lemma}

\theoremstyle{remark}

\AtBeginDocument{%
  }

\makeatletter
\fancypagestyle{firstpageieee}{%
  \fancyhf{}%
  \fancyfoot[C]{%
    \parbox{0.92\textwidth}{%
      \centering\footnotesize
      This work has been submitted to the IEEE for possible publication.
      Copyright may be transferred without notice, after which this version may no longer be accessible.
    }%
  }%
}
\makeatother

\begin{document}

\newcommand{\name}[1]{\textsf{PiLLar}}
\title{PiLLar: Matching for Pivot Table Schema via LLM-guided Monte-Carlo Tree Search}

\author{Yunjun Gao}
\affiliation{
    \institution{Zhejiang University}
}
\email{gaoyj@zju.edu.cn}

\author{Chuangyu Ouyang}
\affiliation{
    \institution{Zhejiang University}
}
\email{cy.ouyang@zju.edu.cn}

\author{Congcong Ge}
\affiliation{
    \institution{Zhejiang University}
}
\email{gcc@zju.edu.cn}

\author{Yifan Zhu}
\affiliation{
    \institution{Zhejiang University}
}
\email{xtf\_z@zju.edu.cn}


\begin{abstract}
Pivot tables are ubiquitous in data lakes of modern data ecosystems, making accurate schema matching over pivot tables a key prerequisite for data integration.
In this paper, we focus on \emph{matching for pivot table schema}, which is a novel joint schema-value matching task.
It aims to align schemas between pivot tables and standard relational tables, where a correct match must be semantically consistent at the schema level and compatible at the value level.
However, due to the inherent data sensitivity of this task, the prevalence of anonymized data in practice poses significant challenges to its matching accuracy and generalization capability.
To tackle these challenges, we propose \name~, the first matching for pivot table schema framework.
We first formulate \name~ as an \emph{LLM-driven search paradigm} that operates with minimal annotated privacy-compliant data, thereby achieving training-free adaptation across diverse domains.
Next, we provide a \emph{theoretical analysis} on the error dynamics of the paradigm to ensure the asymptotic convergence of the proposed method.
Furthermore, we introduce a new benchmark \textsc{PTbench}, derived from \emph{four} representative real-world domains and constructed by mining unpivot-suitable tables, performing unpivot on semantically coherent attributes, and applying sampling and anonymization. Extensive experiments demonstrate the superiority of \name~, which achieves an average accuracy of $87.94\%$ on the correctly predicted matches.
\end{abstract}




\maketitle

\thispagestyle{firstpageieee}

\section{Introduction}
\label{sec:intro}

It is becoming increasingly easier for companies to acquire large amounts of data from diverse sources~\cite{armbrust2020delta, rethinkdata}.
This trend enables SaaS providers (e.g., Salesforce) to deliver richer data analysis capabilities~\cite{powerbi, fabric, crm, looker} by integrating or linking datasets from different sources.
\emph{Schema matching}~\cite{rahm2001survey} serves as a prerequisite for such integration.
It aims to identify the semantic correspondence between attributes across disparate data sources.
Recently, researchers have devoted considerable efforts to schema matching on \emph{standard relational tables}~\cite{doan2000learning, shraga2020adnev, liu2024gram}.
However, they overlook schema matching for \emph{pivot tables}, a task of critical importance given their ubiquity in enterprise reporting and business intelligence tools~\cite{jansen2018use, cho2025data, shadow_it_spreadsheets}.


\noindent \textbf{New task -- matching for pivot table schema.}
In modern data ecosystems, the prevalence of data lakes has led to the proliferation of diverse wild tables~\cite{huang2018auto, wang2019uni}, among which pivot tables are ubiquitous~\cite{jansen2018use}.
Pivot operations transpose values into attribute headers, which obscure semantics and complicate schema matching~\cite{wickham2014tidy}.
Meanwhile, pivot and unpivot are integral reshaping transformations in business intelligence (BI) and machine learning (ML) data preparation pipelines~\cite{yan2020auto, yang2021auto}.
Thus, effectively matching between pivot tables and standard relational tables is essential for facilitating critical tasks such as master data management~\cite{loshin2010master, informatica-mdm} and cross-system data interoperability~\cite{rahm2001survey, wilkinson2016fair}. 
Notably, matching for pivot table schema introduces an additional data quality pitfall beyond conventional schema matching: the choice of the unpivot attribute set implicitly determines the semantics of the generated \textit{Metric--Value} fields, and an incompatible choice can silently alter attribute semantics.
Such subtle errors arising from changes in data semantics can be difficult to detect yet can severely disrupt downstream decision-making and ML pipelines~\cite{schelter2018automating}.
Accordingly, industrial data governance and profiling systems emphasize early validation of both schema evolution and semantic drift to prevent error propagation along data pipelines~\cite{informatica2025cloud,palantir2021trust}.
Therefore, matching for pivot table schema requires a joint schema--value perspective to ensure both schema-level semantic correspondence and value compatibility. Solving this task is non-trivial, and existing approaches fall short. We illustrate the challenges in Example~\ref{example:intro}.


\begin{figure*}[t]
    \includegraphics[width=\linewidth]{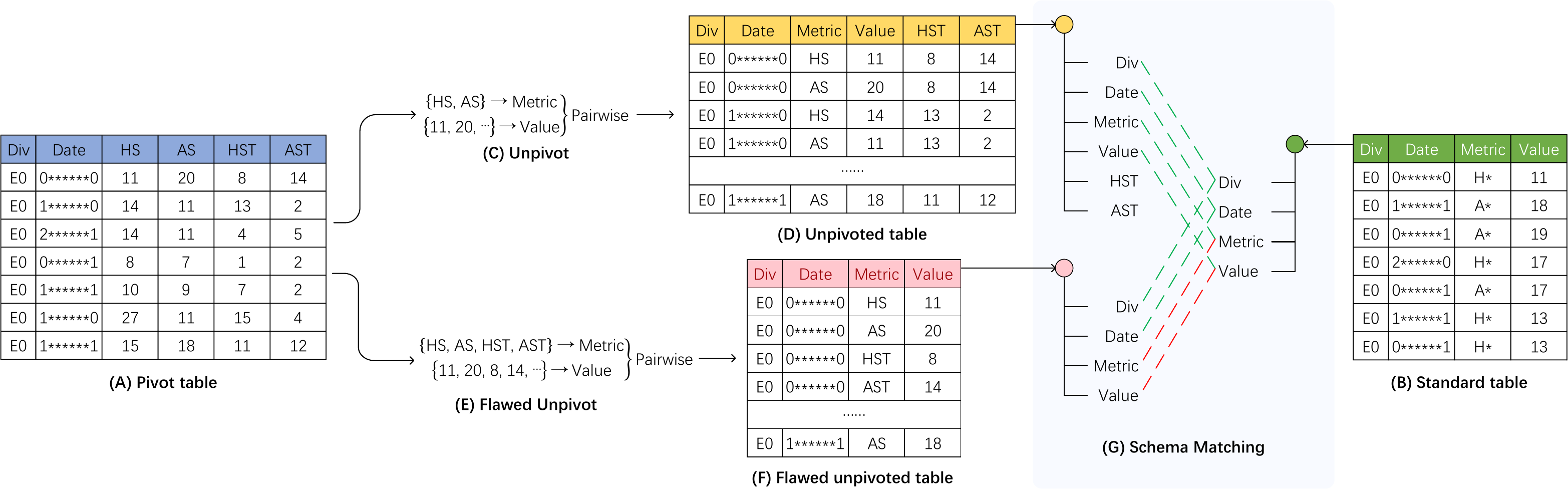}
    \vspace{-6mm}
    \caption{An example of performing matching for pivot table schema with separate unpivot and schema matching steps}
    \label{fig:example}
    \vspace{-3mm}
\end{figure*}

\begin{example}
\label{example:intro}
We consider two datasets from Football-Data~\cite{football-data}, recording information about football match results.
Matching schemas between these datasets enables data integration for downstream team and tactical analysis.
Figure~\ref{fig:example} depicts a snippet.
Figure~\ref{fig:example}(A) is a pivot table $T_l$ that reports, for each match, four shot-related statistics as attributes: home shots (\texttt{HS}), away shots (\texttt{AS}), home shots on target (\texttt{HST}), and away shots on target (\texttt{AST}).
Figure~\ref{fig:example}(B) is a standard relational table $T_r$ that adopts a pairwise \texttt{Metric}--\texttt{Value} schema, where each row corresponds to one metric-value record of a match, facilitating aggregation and comparison.
In addition, data records of \texttt{Date} and \texttt{Metric} attributes are anonymized to protect sensitive information. Specifically, the \texttt{Date} values may reveal schedules of football teams; the \texttt{Metric} values often encode proprietary performance indicators used by clubs or analytics providers.
To correctly align $T_l$ with $T_r$, one must first identify the \emph{value-compatible} unpivot attribute set $\{\texttt{HS}, \texttt{AS}\}$ and unpivot $T_l$ into the intended shot-count records shown in Figure~\ref{fig:example}(D). Schema matching is then performed between the unpivoted table $T'_l$ and $T_r$, obtaining the correct matches illustrated in Figure~\ref{fig:example}(G).
\end{example}
\vspace{-1mm}


\vspace{0.1in}
\textbf{Challenge I:} \textit{How to guarantee the matching for pivot table schema accuracy with anonymous data?}
In practice, schema matching typically occurs in cross-departmental collaboration~\cite{nadal2022operationalizing}.
In this scenario, data providers can access table metadata (i.e., attributes) but are denied access to the actual data records due to privacy policy, wherein sensitive data is often anonymized.
Conventional schema matching methods rely primarily on features derived from attributes themselves.
Data records are optional for providing external knowledge~\cite{do2002coma, zhang2023schema, liu2024gram}.
In contrast, matching for pivot table schema highly demands access to data records to identify the value-compatible unpivot attribute set. For example, in Figure~\ref{fig:example}, the candidate set $\{\texttt{HS}, \texttt{AS}, \texttt{HST}, \texttt{AST}\}$ is reasonable if one only considers the schema of $T_l$. However, this selection would lead to a resulting unpivoted table (as shown in Figure~\ref{fig:example}(F)) that conflates distinct semantics, namely shots and shots on target, making the generated fields semantically inconsistent with the standard table $T_r$. Under such semantic drift, a matcher may still produce seemingly plausible matching at the schema level, but it becomes incorrect at the value level.
This ambiguity can only be resolved through record-level data patterns.
Such a strong requirement for data records conflicts with the recently growing privacy concerns of data providers in practice~\cite{regulationeu, ccpa}.
Therefore, our work focuses on matching schemas under the constraint of \textit{data minimization}, which is a widely accepted concept in data protection regulations and commercial systems~\cite{regulationeu, gdpr, palantir2024privacy}, to achieve a deliberate balance between privacy preservation and matching accuracy.
\vspace{0.1in}
\textbf{Challenge II:} \textit{How to effectively perform matching for pivot table schema?}
The dual demands of generality and data privacy make Large Language Models (LLMs) particularly suitable for this task. LLMs with billions of parameters exhibit strong zero-shot/few-shot generalization capabilities~\cite{brown2020language}, which facilitates matching for pivot table schema across diverse domains without requiring sensitive data or intensive training. 
Despite this, their inherent instability means that even the most advanced LLMs are not a reliable standalone solution.
We still take Figure~\ref{fig:example} as an example.
Current schema matching approaches cannot directly match these attribute groups, as they ignore the transformations of schema structure.
A straightforward solution to perform matching for pivot table schema on this example is to first unpivot $T_l$ and then perform conventional schema matching.
Yet $T_l$ contains a candidate set of attributes that are probably to be unpivoted, namely $\mathcal{A_{\text{cand.}}}=\{\texttt{HS}, \texttt{AS}, \texttt{HST}, \texttt{AST}\}$.
Experiments show that both open-source and closed-source state-of-the-art LLMs generate the unpivot result $\mathcal{A}_\text{unpivot} = \{\texttt{HS}, \texttt{AS}, \texttt{HST}, \texttt{AST}\}$  (detailed results are presented in Appendix~\ref{app:case_study}).
This result is reasonable when querying only for attributes unpivotable in the schema of $T_l$, but it is incorrect in our task since we aim to obtain correct matching results between the two input tables.
Although extensive SOTA LLM-driven approaches have been devoted to schema matching, matching schemas on these flawed input tables can only amplify the error, as described in Example~\ref{example:intro}.
Hence, the core challenge is to mitigate error propagation by jointly ensuring a value-compatible unpivot attribute set choice and a verifiable schema matching to the standard table.
In light of these challenges, we make the following contributions:


\begin{itemize}[leftmargin=*]
    \item \textit{Flexible Framework}.
    We propose \name~, the first LLM-guided search framework for matching pivot table schemas, requiring only few-shot labeled anonymized data and enabling training-free adaptation across domains.
    \item \textit{Convergent Search Paradigm}.
    We formulate the task as a bounded-stochastic search to mitigate LLM instability, guaranteeing both exploration completeness and search efficiency. We further provide a formal analysis establishing asymptotic convergence.
    \item \textit{Self-correcting Iterative Search Strategy}.
    We design an \textit{identifier–-judger} iteration in which robust prompts drive the identifier to propose candidate unpivot attribute set, and a multi-dimensional validator serving as a judger provides per-iteration feedback that steers the identifier in the next iteration, mitigating unreliable unpivoting and error propagation during search.
    \item \textit{Extensive Experiments}.
    We propose a new benchmark from four real-world domains and show the effectiveness of \name~, which achieves an average accuracy of $87.94\%$ on the correctly predicted matches and $94.45\%$ on the correctly operated attributes.
\end{itemize}

\section{Problem Statement}
\label{sec:preliminaries}




Given a relational table $T = (\mathcal{A}, \mathcal{D}, \mathcal{V})$, where $\mathcal{A}$ refers to the attributes, $\mathcal{D}$ refers to the description of attributes, and $\mathcal{V}$ refers to the sampled anonymized data records.
In this paper, we focus on the task of matching for pivot table schema.
Let $T_l = (\mathcal{A}_l, \mathcal{D}_l, \mathcal{V}_l)$ and  $T_r = (\mathcal{A}_r, \mathcal{D}_r, \mathcal{V}_r)$ be the two given tables.
The objective is to first identify the unpivot operator $\Phi = (\mathcal{A}_\text{unpivot}, A_\text{var}, A_\text{value})$ that can transform $T_l$ into $T'_l$ to align with $T_r$, where $\mathcal{A}_\text{unpivot}$ represents unpivot attribute set, $A_\text{var}$ represents the attribute name derived from the unpivoted attribute names, and $A_\text{value}$ represents the attribute name derived from the corresponding data records. The matching result $\pi$ is then generated between the two attribute sets.
Here, we consider one-to-one matches where a match specifies that the two attributes are equal to each other, which is a common assumption in schema matching~\cite{doan2000learning, zhang2023schema, liu2024gram, seedat2024matchmaker, zhang2025smutf}.
Formally,
\begin{align}
\pi : & ~\mathcal{A}'_l \to \mathcal{A}_r \cup \{\text{Null}\} \\
\mathrm{s.t.} & ~\forall A_i, A_j \in \mathcal{A}'_l, ~\pi(A_i) \neq \pi(A_j) ~\lor \pi(A_i) = \text{Null} \nonumber
\end{align}

\section{PiLLar Framework}
\label{sec:framework}

In this section, we describe the framework of \name~ in detail.
Since we formulate matching for pivot table schema as a search problem, we first introduce the proposed \emph{search paradigm}; we then detail the \emph{self-correcting iterative search strategy} of \name~.

\subsection{Search Paradigm}
\label{sec:paradigm}

Recall that LLM is a powerful tool for matching for pivot table schema. However, due to the inherent hallucination problem of LLMs~\cite{farquhar2024detecting}, it is unsafe to rely solely on LLM generation, as described in Section~\ref{sec:intro}.
Considering that \emph{search paradigm} can effectively mitigate the problem of losing correct answers caused by LLM's uncertainty~\cite{shorinwa2025survey}, we would like to formulate the task of matching for pivot table schema as an LLM-guided search problem.
For matching for pivot table schema, identifying the attributes to be unpivoted is a necessary step; however, exploring the complete search space of all attribute subsets is computationally prohibitive, with a complexity of $O(2^n)$. Motivated by the effectiveness of MCTS in balancing the accuracy and efficiency in search problems~\cite{browne2012survey}, we propose a bounded-stochastic MCTS variant guided by the LLM. To ensure theoretical convergence—and thereby mitigate potential hallucination from the LLM—we incorporate a \emph{bounded-stochastic expansion strategy} into the search process. Detailed theoretical analysis can be found in Section~\ref{sec:analysis}.


\subsubsection{Overview of the Search Process}
\label{sec:search_overview}

~

We first outline how the proposed bounded-stochastic MCTS variant is involved in our framework. Starting from an initial candidate unpivot attribute set, the MCTS repeatedly executes four phases in each iteration: selection, expansion, evaluation and backpropagation. During the selection phase, the search process traverses the current search tree by applying a UCT policy to identify the most promising node to be expanded. In the expansion phase, a new candidate set is generated by our proposed bound-stochastic expansion policy. Once generated, the evaluation phase performs schema matching between the unpivoted table and the standard table, producing a deterministic reward. Finally, in the backpropagation phase, this reward is propagated along the visited path, which gradually biases the search toward high-quality candidates. Detailed implementation of this process is presented in Section~\ref{sec:realization}.
Then, we detail the design of our proposed search paradigm, together with a theoretical convergence analysis.

\subsubsection{Search and Update Mechanism}
\label{sec:mechanism}

~

\textbf{Bounded Stochastic Strategy Design.}
To ensure theoretical convergence while maintaining LLM guidance, we adopt a bounded-stochastic strategy in the expansion phase. During expansion, a new child is produced either by an LLM-guided generation with probability $1 - \varepsilon$, or by a radius-$1$ random modification (i.e., a single add/remove/swap operation of one attribute) with probability $\varepsilon>0$. Each set generated by the random modification is prohibited from being regenerated by it. This mechanism guarantees that all feasible candidates are theoretically reachable (formal proof can be found in Appendix~\ref{app:proofs}). Combined with the LLM-guided generation, \name~ reaches a balance between directed reasoning and theoretical search completeness, enabling the search to cover potential candidate unpivot attribute sets without exhaustive enumeration.


\textbf{Node Evaluation and Reward Propagation.} For each candidate attribute set, we generate a corresponding schema matching result and a quantized reward. To ensure stability and analytical tractability, \name~ adopts a deterministic reward formulation. The reward of each node is designed to be bounded and noise-free, ensuring that every evaluation consistently reflects the true quality of the node. Formally, we assume the reward $R(v)$ of a node $v$ satisfies $R(v) \in [0, \Omega(T_l, T_r)], \forall v$, where $\Omega(T_l, T_r)$ is a deterministic upper bound of the reward based solely on the given input tables.

Once a node is evaluated, its estimated reward is updated through the \textbf{max–average backpropagation rule}:
\begin{equation}
    \label{eq:maxavg}
    Q_{v_i} \leftarrow \frac{1}{2} \left(Q_{v_i} + \max \left\{ R_{v_i}, \max_{v_j \in \mathrm{children}(v_i)} Q_{v_j} \right\} \right)
\end{equation}
where $Q_{v_i}$ represents the estimated value for node $v_i$, and $R_{v_i}$ represents the reward of node $v_i$.  
This update design reflects the optimization-oriented role of our framework: instead of using simulated rollouts to estimate the value of a node as in traditional MCTS, each node’s $Q$-value in our framework measures the potential improvement obtainable by further exploration from that node. Therefore, only better descendants are allowed to update their ancestors through the $\max$ operator, guiding the search toward regions of higher potential.

Together, the $\varepsilon$-random expansion and the max–average backpropagation construct a bounded-stochastic MCTS variant that achieves asymptotic search completeness, while preserving the efficiency advantages of LLM-guided reasoning.

\begin{figure*}[t]
    \includegraphics[width=\linewidth]{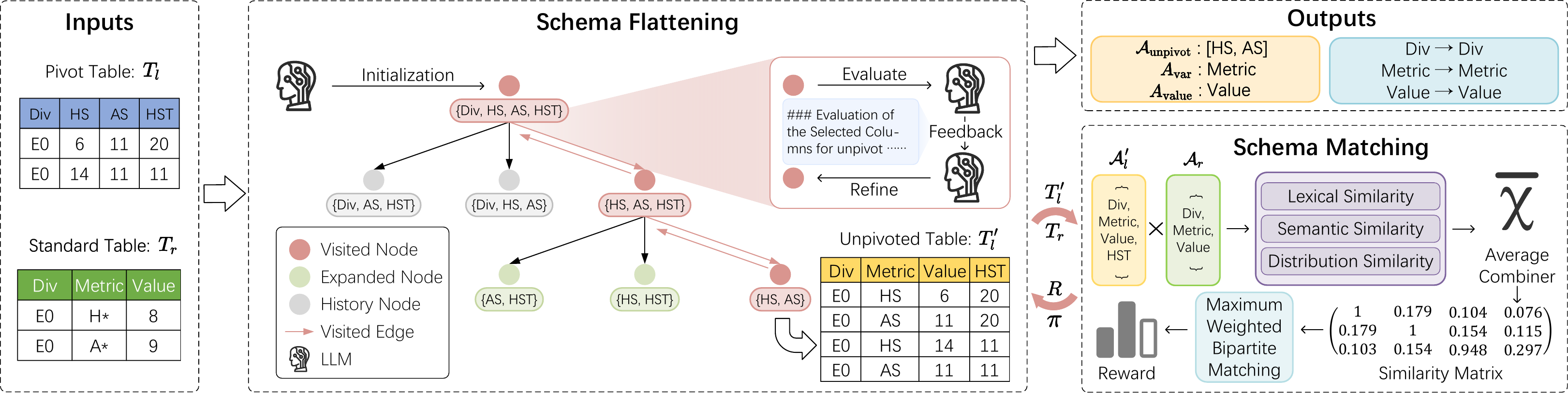}
    \caption{Overview of the \name~ framework}
    \label{fig:overview}
    \vspace{-2mm}
\end{figure*}

\subsubsection{Theoretical Analysis}
\label{sec:analysis}

To understand how bounded stochasticity influences convergence, we analyze the error dynamics, including weak-hit disturbance, a single-hit contraction kernel, and an asymptotic block recursion.

\vspace{0.1in}
\textbf{Analysis Setup.}
Let $\mathcal{S}$ be the finite candidate unpivot attribute sets (defined by the finite attribute universe). Each search-tree node $v$ encodes a candidate $s=\tau(v) \in \mathcal{S}$ (where $\tau$ is many-to-one since Self-Refine may generate distinct nodes for the same $s$). For each node $v$, its deterministic reward is denoted by $R(v) \in [0, \Omega(T_l, T_r)]$. Suppose there exists an optimal node $v^\star$ corresponding to the optimal candidate $s^\star = \tau(v^\star)$, such that $R(v^\star) = \max_{v} R(v)$. Let $(v_0, \ldots, v_H)$ denote a minimal witnessing path from the root (depth $0$) to $v^\star$ (depth $0 \leq H \leq \lvert \mathcal{S} \rvert$). Let $e_d$ be the absolute error at depth $d$ with respect to $R(v^\star)$ before an iteration’s backpropagation, and $e'_d$ the error after that backpropagation in the same iteration.

\begin{proof}[Sketch]
We sketch the argument and defer details to Appendix~\ref{app:proofs}.
Since $\mathcal{S}$ is finite and expansion is $\varepsilon$-randomized, every feasible candidate is generated with probability~$1$ (probabilistic completeness). Once the optimal node $v^\star$ is discovered, the max--average backup yields a contraction effect on the root's estimation error whenever $v^\star$ is reached and backpropagated. Moreover, under UCT with deterministic rewards, suboptimal selections become asymptotically negligible, so the disturbance from weak updates vanishes. Therefore the root estimate converges asymptotically to $R(v^\star)$.
\end{proof}

\subsection{Self-correcting Iterative Search Strategy}
\label{sec:realization}

Based on the bounded stochastic search paradigm described in Section~\ref{sec:paradigm}, we now detail how to perform \name~ in an iterative manner via two key components, i.e., \emph{schema flattening} and \emph{schema matching}.
Figure \ref{fig:overview} depicts an overview.

\subsubsection{Schema Flattening}
This component serves as an \emph{identifier}, which aims to flatten the schema of the input pivot table $T_l$ into a standard format that complies with the input standard table $T_r$.
It leverages the proposed LLM-guided MCTS variant to explore candidate unpivot operators.
To better leverage the semantic capability of LLMs, we adopt the Self-Refine~\cite{madaan2023self} mechanism to optimize the generated operators. Self-Refine provides a feedback-controlled update mechanism that turns unguided exploration into a directional process that incrementally improves candidate quality. Besides, this mechanism can integrate well with MCTS, as the tree structure naturally preserves refinement paths while maintaining candidate diversity through branching exploration. Under this design, each node on the Monte-Carlo tree represents a candidate unpivot operator, and each edge represents a Self-Refine/random radius-$1$ modification process.

Schema flattening consists of four phases, namely \textit{initialization}, \textit{selection}, \textit{expansion} and \textit{backpropagation}.
The detailed description of the backpropagation phase can be found in Section~\ref{sec:mechanism}.
The details of initialization, selection, and expansion are described below.

\textbf{Initialization.} It acts as the start of the \name~ framework. During this phase, the root node of the Monte-Carlo tree is generated by querying LLM for an initial candidate unpivot attribute set.

\begin{figure}[t]
    \centering
    \includegraphics[width=\linewidth]{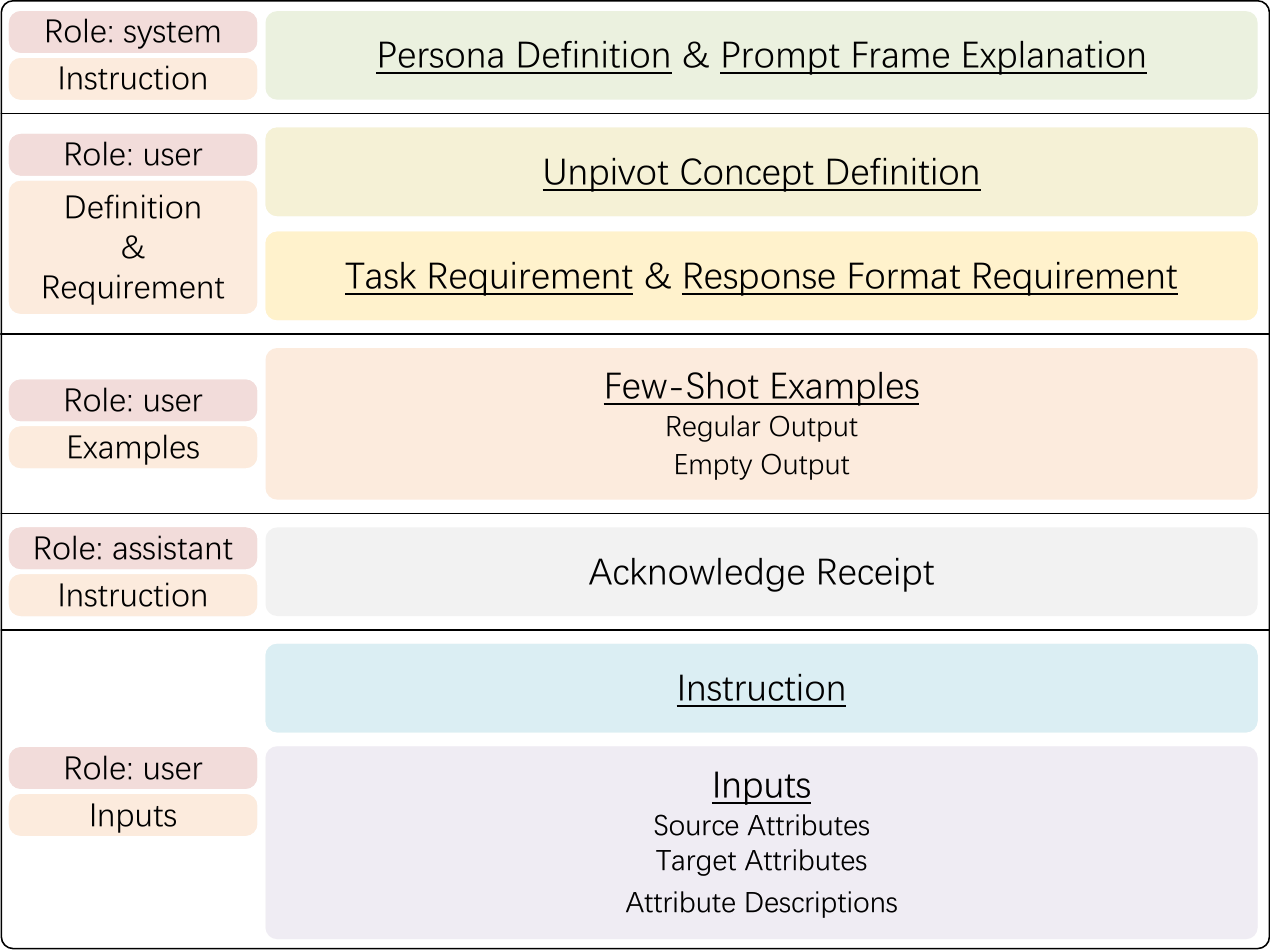}
    \vspace{-5mm}
    \caption{An example of an initialization prompt template}
    \label{fig:init_prompt}
    \vspace{-4mm}
\end{figure}

Motivated by the \textit{clear layout} prompt style~\cite{dong2023c3}, we design a standardized initialization prompt template to facilitate this generation, as shown in Figure~\ref{fig:init_prompt}. The prompt begins with a persona definition and a prompt frame explanation, which are system-level instructions that define the model's persona and the frame of the following prompt.
This instruction ensures domain alignment and establishes a consistent interpretive frame before specific requirements and input are provided. 
This is followed by the definition and requirement section, which clarifies the unpivot concept definition, introduces the task of detecting unpivotable attributes in task requirement and specifies the output in response format requirement.
Here, LLM is required to select more attributes within reason so that in the \textbf{expansion} phase we can ask LLM to unidirectionally reduce the attributes in the candidate unpivot attribute set, and thus avoid getting stuck in a cycle.
After that, we provide two illustrative \emph{few-shot examples} that serve as few-shot demonstrations to help the model understand the task definition and output pattern. The first example presents a typical unpivot scenario, clarifying the task semantics and indicating the correct output format, while the second specifies the edge case where no attribute should be unpivoted. We empirically found that these examples are sufficient for conveying the task intent, i.e., to propose a reasonably inclusive set of potentially unpivotable attributes, narrowing and structuring the search space for later refinement. The challenge of identifying the unpivot attribute set is resolved during the subsequent \textbf{expansion} phase rather than through additional examples at initialization.
In the last section, we provide the input information with sharp symbols to make the prompt clearer.

\textbf{Selection.} The selection phase is the beginning of each iteration after initialization. Starting from the root node, this phase traverses the search tree by visiting child nodes until reaching a node not \textit{fully expanded}. Since each child node represents a refinement of the current unpivot attribute set, a node actually has infinite expansion space.
Therefore, we define a node as fully expanded when it reaches a maximum number of child nodes, where the maximum number is a user-defined hyper-parameter.
When the selection is currently performed at node $v_i$, each child node $v_j \in \mathcal{C}(v_i)$ is assigned a UCT~\cite{kocsis2006bandit} score:
\begin{equation}
    \text{UCT}_{v_j} = Q_{v_j} + C \sqrt{\frac{\ln N_{v_i}}{N_{v_j} + \epsilon}} \label{eq:uct}
\end{equation}
where $Q_{v_j}$ represents the cumulative reward of node $v_j$, $N_{v_i}$ represents the visit count of $v_i$, $N_{v_j}$ represents the visit count of $v_j$, $C$ is a hyperparameter to balance exploitation and exploration, $\epsilon$ is a small constant preventing Equation \eqref{eq:uct} from dividing zero. The child node with the greatest UCT score is selected.


\textbf{Expansion.} This phase generates child nodes for the selected node, similar to the expansion phase in traditional MCTS. The difference is that it does not randomly generate a child node in the possible search space, but follows a combined expansion strategy. With a probability of $1 - \varepsilon$, it asks LLM to optimize the current unpivot attribute set, and with a probability of $\varepsilon$ it performs a radius-$1$ random modification. This strategy avoids inefficient blind exploration by leveraging semantic guidance from LLM while still preserving sufficient randomness for convergence. For a node $v$ to be expanded, we have
\begin{empheq}[left={\eqmakebox[left][r]{$fb=$$\phantom{\empheqlbrace}$}}]{equation}
    \eqmakebox[right][l]{$\mathcal{M}(p_\text{fb} \| T_l \| T_r \| v)$} \quad \eqmakebox[cond][l]{} \label{eq:feedback}
\end{empheq}
\begin{empheq}[left={\eqmakebox[left][r]{$v'=$}}\empheqlbrace]{alignat=2}
    &\eqmakebox[right][l]{$\mathcal{M}(p_\text{init} \| T_l \| T_r \| v \| p_\text{refine} \| fb)$} &\quad& \eqmakebox[cond][r]{$\text{w.p. } 1 - \varepsilon$} \label{eq:refine} \\
    &\eqmakebox[right][l]{$\mathcal{N}(v)$} &\quad& \eqmakebox[cond][l]{$\text{w.p. } \varepsilon$}
\end{empheq}
where $p_\text{init}$ refers to the prompt that guides model $\mathcal{M}$ to generate the initial unpivot attribute set, $p_\text{fb}$ refers to the prompt that guides $\mathcal{M}$ to generate the feedback $fb$ for $v$, $p_\text{refine}$ refers to the prompt that guides $\mathcal{M}$ to generate optimized output $v'$, and $\mathcal{N}(v)$ refers to a random neighbor obtained by a radius-$1$ modification of the candidate unpivot attribute set of node $v$. Although the use of iterative LLM querying here increases the time cost compared to traditional schema matching approaches, this cost remains acceptable in real-world schema matching scenarios. Industrial studies~\cite{iovine2025effective} have reported that Amazon's large-scale schema matching systems, which also rely on LLM reasoning for attribute alignment, typically operate on an hour-level timescale while still achieving substantial efficiency gains, reducing human review time by more than 90\%.


Equation \eqref{eq:feedback} represents the process of model $\mathcal{M}$ evaluating the current unpivot attribute set and offering feedback.
We design the following three \emph{calibration hints} in the prompt for this query to calibrate some biases we found empirically:

\begin{enumerate}[leftmargin=*]
    \item The evaluation should focus on the transformation between the input tables using the current unpivot attribute set. Sometimes LLM may judge the design of input tables which is not what we expect here, so we use this hint to calibrate it.
    \item The optimization task should focus on reducing the size of the current unpivot attribute set. This hint is designed to coordinate with $p_\text{init}$ in order to build a unidirectional reasoning path, as we mentioned in the \textbf{initialization} phase.
    \item All attributes mentioned in the feedback should be selected from the attribute sets $\mathcal{A}_l$ and $\mathcal{A}_r$. This hint is used to prevent $\mathcal{M}$ from generating feedback that contains attributes that do not exist and leads to an illegal unpivot attribute set.
\end{enumerate}


Equation \eqref{eq:refine} represents the process of $\mathcal{M}$ refining the current unpivot attribute set based on the feedback generated from Equation \eqref{eq:feedback}.
The prompt template for this step produces a multi-turn conversation by sequentially incorporating refinement requirements and feedback after the initialization prompt and the LLM’s response.
Based on the bias we discovered in experiments that even if the feedback indicates the selection is ideal, LLM may still modify the unpivot attribute set according to the analysis procedure in the feedback, we calibrate LLM with the \emph{calibration hint} that it can leave the unpivot attribute set unchanged under this circumstance.

After the new candidate unpivot attribute set is generated, we can query LLM for a pair of corresponding attribute names $A_\text{var}$ and $A_\text{value}$ derived from the unpivot attributes to obtain an unpivot operator, and apply the operator to the pivot table $T_l$  to get the unpivoted table $T'_l$. Evaluation of this operator is performed in the schema matching component, for which we provide a detailed illustration in Section~\ref{sec:sm}.




\subsubsection{Schema Matching}
\label{sec:sm}
The schema matching component acts as the \emph{judger}. It aligns the unpivoted table $T'_l$ with $T_r$ and evaluates the reward of the unpivot operation. This component operates in two phases: \textit{Similarity Calculation} and \textit{Matching Generation}.

\textbf{Similarity Calculation.} During this phase, we compute the similarity score between each pair of attributes and obtain a similarity matrix between $T'_l$ and $T_r$. To comprehensively measure attribute correspondence, we adopt a \emph{multi-dimensional evaluation metric} that integrates three complementary dimensions: (i) a lexical signal using \emph{Levenshtein distance}~\cite{levenshtein1966binary} on attribute names, (ii) a semantic signal using the \emph{cosine similarity of embeddings}~\cite{gomaa2013survey} of attribute names, and (iii) a distributional signal based on \emph{Jensen--Shannon (JS) divergence}~\cite{lin2002divergence} between the value distributions of the two attributes. Each dimension outputs a similarity score in $[0,1]$.

For each attribute pair, all three similarity scores are computed when both attributes are integer-valued, since the JS divergence applies solely to numerical distributions. Otherwise, only lexical and semantic similarities are used for computation.
The similarity of the value distribution with limited precision from sampled data records acts as complementary evidence for the metric to evaluate the plausibility of unpivot attributes $\mathcal{A}_\text{unpivot}$. This design complements the LLM-based unpivot identification in the schema flattening component, enabling \name~ to achieve a balanced selection that aligns both semantic consistency and data distribution. In order to combine these scores into a comprehensive similarity score, we apply a combiner that takes these scores as input and outputs a similarity score between 0 and 1. In our current implementation, we use a simple average combiner based on the experiments (detailed experimental results can be found in Appendix~\ref{app:similarity_metrics}). 

\textbf{Matching Generation.} A similarity matrix $M$ is formed after calculating all pairs of attributes between $T'_l$ and $T_r$. To get a schema matching and a quantized evaluation, we apply maximum weighted bipartite matching to the matrix with a modified Jonker-Volgenant algorithm without initialization~\cite{crouse2016implementing}. The maximum reward and the corresponding matching are sent back to the schema flattening component and guide the subsequent iterations.

\section{Experiments}
\label{sec:experiments}

In this section, we conduct extensive experiments to evaluate the effectiveness of \name~ using our benchmark datasets.

\subsection{Experimental Settings}
\label{sec:exp_settings}

~

\textbf{Datasets.}
To the best of our knowledge, there is no widely-acknowledged benchmark dataset for assessing schema matching over pivoted tables.
To study the performance of \name~ in real-world scenarios, we propose a new benchmark named \textsc{PTbench}, which contains four datasets, using real cases from two categories: (i) online user forums and (ii) spreadsheet-tables from real-life ETL processes~\cite{li2023auto}.
Table \ref{table:statistics} shows the statistics of \textsc{PTbench}.

\underline{Adult.}
This dataset is extracted from the 1994 Census Bureau database~\cite{censusbureau}.
It contains data relevant to demographic information and economic conditions. Both the pivot table $T_l$ and standard table $T_r$ contain 19 attributes, and the unpivot attribute set contains 2 attributes. 19 pairs of attributes are matched between the two tables.

\underline{Football.}
This is the Premier League dataset published in Football-data~\cite{football-data}.
It contains data such as team information, goals and shots, etc. The pivot table consists of 23 attributes, among which 6 form the unpivot attribute set, while the standard table contains 13 attributes. There are 13 attribute matches between the two tables.

\underline{President.}
This dataset is extracted from real-life ETL processes. It contains evaluation data from various perspectives on presidents of the USA including war record, economic approval rate and so on. The pivot table contains 12 attributes and the standard table contains 4 attributes. The unpivot attribute set consists of 5 attributes, and the ground truth contains 4 pairs of matches.


\underline{Gene.}
This dataset is derived from the GTEx v11 public sample annotation table~\cite{GTExPortal}. It contains rich metadata for human tissue samples such as tissue type, detailed tissue subtype and RNA quality metrics. The pivot table contains 119 attributes and the standard table contains 96 attributes. The unpivot attribute set consists of 25 attributes, and the ground truth contains 96 pairs of matches.

All four tables have been stratified sampled to balance the requirements between user privacy and data distribution. In addition, we apply anonymization to string attributes to protect user privacy. 

\begin{table}[t]\small
    \centering
    \caption{Statistics of \textsc{PTbench} used in experiments}
    \label{table:statistics}
    \begin{tabularx}{\linewidth}{cc|>{\centering\arraybackslash}X|>{\centering\arraybackslash}X}
        \toprule
        \multicolumn{2}{c|}{Dataset} & \makecell[c]{\# Total Entities} & \# Attributes\\ \hline
        \multirow{2}{*}{Adult} & pivot table & 32,561 & 19 \\
        ~ & standard table & 65,122 & 19 \\ \hline
        \multirow{2}{*}{Football} & pivot table & 380 & 23 \\
        ~ & standard table & 2,280 & 13 \\ \hline
        \multirow{2}{*}{President} & pivot table & 43 & 12 \\
        ~ & standard table & 215 & 4 \\ \hline
        \multirow{2}{*}{Gene} & pivot table & 48,231 & 119 \\
        ~ & standard table & 1,205,775 & 96 \\
        \bottomrule
    \end{tabularx}
    \vspace{-4mm}
\end{table}


\vspace{0.2em}
\textbf{Comparative Approaches.}
We compare \name~ against the following representative approaches: COMA 3.0~\cite{massmann2011evolution}, DisB~\cite{zhang2011automatic}, GRAM~\cite{liu2024gram} and NaiveP (a naive pipeline that performs unpivot attribute identification and schema matching independently without iterative refinement). See Appendix~\ref{app:baseline_details} for more details.

\textbf{Evaluation Metrics.}
The end-to-end accuracy is the traditional evaluation metric for the schema matching approaches. Since our datasets contain unpivot attributes, this metric cannot evaluate the performance of the approaches, so we additionally adopt per-attribute accuracy as a second metric. Specifically, we use the following two metrics:

\begin{itemize}[leftmargin=*]
    \item \textbf{\textit{End-to-End Accuracy} ($\boldsymbol{Acc}_\text{E2E}$).} $Acc_\text{E2E}$ represents the ratio of correctly predicted matches to ground truth matches, evaluating the performance of approaches end-to-end. For the attributes generated by the unpivot operation, a match is counted as correct only if the approach correctly identifies the unpivot attribute set and matches the generated attributes with the correct target.
    \item \textbf{\textit{Per-Attribute Accuracy} ($\boldsymbol{Acc}_\text{per\_attr.}$).} $Acc_\text{per\_attr.}$ captures the ratio of \textit{correctly operated} attributes to all attributes. For attributes in the unpivot attribute set, correctly operated means successfully recognizing them as the attributes to be unpivoted, while for other attributes, it means matching them to the correct target.
\end{itemize}

Formally, for input attribute sets $\mathcal{A}_l$ and $\mathcal{A}_r$, the ground truth unpivot attribute set is $\mathcal{A}_\text{unpivot}$ and the ground truth matching is $\pi$. If the approach identifies an unpivot attribute set $\mathcal{A}'_\text{unpivot}$, transform $\mathcal{A}_l$ into $\mathcal{A}'_l$, and generate a matching $\pi'$, we have:
\allowdisplaybreaks
\begin{align}
    & Acc_\text{E2E} = \frac{\bigg| \left\{ A \ | \ A \in \mathcal{A}'_l \land \pi(A) = \pi'(A) \neq \text{Null} \right\}\bigg|}{\bigg| \left\{ A \ | \ A \in \mathcal{A}'_l \land \pi(A) \neq \text{Null} \right\}\bigg|} \\
    & Acc_\text{per\_attr.} = \frac{\lvert \mathcal{A}_\text{correct} \rvert}{\lvert \mathcal{A}_l \rvert + \lvert \mathcal{A}_r \rvert}
\end{align}
where
\begin{equation}
\begin{aligned}
    \mathcal{A}_\text{correct} = & \left\{ A \ | \ A \in (\mathcal{A}_l - \mathcal{A}_\text{unpivot}) \land \pi(A) = \pi'(A)\right\} \\
    & \cup \left\{ A \ | \ A \in \mathcal{A}_r \land \exists A' \in \mathcal{A}'_l, \pi(A') = \pi'(A') = A \right\} \\
    & \cup (\mathcal{A}_\text{unpivot} \cap \mathcal{A}'_\text{unpivot})
\end{aligned}
\end{equation}

\subsection{Overall Performance}

\begin{table*}[t]\small
    \centering
    \caption{$Acc_\text{E2E}$(\%) and $Acc_\text{per\_attr.}$(\%) of various approaches on different datasets}
    \vspace{-3mm}
    \label{table:baseline}
    \renewcommand\arraystretch{1.1}
    \begin{tabularx}{\linewidth}{c*{5}{|>{\centering\arraybackslash}X>{\centering\arraybackslash}X}}
        \toprule
        \multirow{2}*{\textbf{Methods}} & \multicolumn{2}{c|}{\textbf{Adult}} & \multicolumn{2}{c|}{\textbf{Football}} & \multicolumn{2}{c|}{\textbf{President}} & \multicolumn{2}{c|}{\textbf{Gene}} & \multicolumn{2}{c}{\textbf{Average}} \\ \cline{2-11}
        ~ & $Acc_\text{E2E}$ & $Acc_\text{per\_attr.}$ & $Acc_\text{E2E}$ & $Acc_\text{per\_attr.}$ & $Acc_\text{E2E}$ & $Acc_\text{per\_attr.}$ & $Acc_\text{E2E}$ & $Acc_\text{per\_attr.}$ & $Acc_\text{E2E}$ & $Acc_\text{per\_attr.}$ \\ \hline
        \multicolumn{11}{c}{All attributes} \\ \hline
        COMA 3.0 & 78.94 & 78.94 & 84.62 & 77.78 & 50.00 & 50.00 & \underline{97.92} & \underline{87.44} & \underline{77.87} & 73.54 \\
        DisB & 56.84 & 56.83 & 35.38 & 37.78 & 30.00 & 47.14 & 12.50 & 11.16 & 33.68 & 38.23 \\
        GRAM & 78.94 & 78.94 & 86.42 & 77.78 & 50.00 & 56.25 & 93.75 & 83.72 & 76.83 & 74.17 \\
        
        NaiveP & \underline{100.00} & \underline{100.00} & \underline{84.62} & \underline{83.33} & \underline{50.00} & \underline{68.75} & 69.58 & 73.86 & 76.05 & \underline{81.49} \\
        \hline
        \textbf{\name~} & \textbf{100.00} & \textbf{100.00} & \textbf{93.85} & \textbf{96.11} & \textbf{60.00} & \textbf{88.75} & \textbf{97.92} & \textbf{92.93} & \textbf{87.94} & \textbf{94.45} \\ \hline
        \multicolumn{11}{c}{Without unpivot attributes} \\ \hline
        COMA 3.0 & 88.23 & 88.23 & 100.00 & 100.00 & 100.00 & 100.00 & \underline{100.00} & \underline{100.00} & \underline{97.06} & \underline{97.06} \\
        DisB & 70.59 & 70.59 & 36.36 & 36.36 & 100.00 & 100.00 & 12.77 & 12.77 & 54.93 & 54.93 \\
        GRAM & 88.23 & 88.23 & 100.00 & 100.00 & 100.00 & 100.00 & 95.74 & 95.74 & 95.99 & 95.99 \\
        NaiveP & \underline{100.00} & \underline{100.00} & \underline{100.00} & \underline{100.00} & \underline{100.00} & \underline{100.00} & 71.06 & 71.06 & 92.77 & 92.77 \\
        \hline
        \textbf{\name~} & \textbf{100.00} & \textbf{100.00} & \textbf{100.00} & \textbf{100.00} & \textbf{100.00} & \textbf{100.00} & \textbf{100.00} & \textbf{100.00} & \textbf{100.00} & \textbf{100.00} \\
        \bottomrule
    \end{tabularx}
    \vspace{-5mm}
\end{table*}

We first conduct a comprehensive comparison of various schema matching methods.
Table \ref{table:baseline} reports the overall performance. This table consists of two parts. 
The first part shows the performance of approaches considering all attributes. 
Since all the competitors are not designed for pivot tables, \name~ significantly outperforms all baselines on all datasets.
Compared to the strongest baseline (NaiveP), it achieves improvements of 15.63\% and 15.90\% in $Acc_\text{E2E}$ and $Acc_\text{per\_attr.}$, respectively. The noticeably lower performance of NaiveP demonstrates that the absence of iterative feedback limits matching effectiveness, whereas \name~ benefits from refinement and cross-component interaction. 
In particular, we observe that \name~ achieves $100\%$ accuracy on both the Adult dataset, while the performance on the President dataset is relatively low. This is because the difficulty in the four datasets lies in accurately identifying the unpivot attribute set.
For Adult, the unpivot attributes are structurally explicit and semantically distinguishable from the rest of the schema, making the optimal unpivot subset almost uniquely identifiable. In contrast, the President dataset contains $10$ similar president evaluation metrics, and semantically many subsets of these $10$ attributes are plausible candidates. This ambiguity results in vast exploration of suboptimal branches, making it harder to reach the ground-truth unpivot attribute set.
The second part shows the accuracy of attributes without unpivot-related ones, i.e., attributes to be unpivoted in the source table and the attributes to be generated in the target table. 
Since the core challenge of the datasets lies in identifying the correct unpivot attributes, once the unpivot attributes are correctly removed, most approaches can achieve strong matching accuracy. In contrast, the DisB method performs the worst because it relies solely on statistical distributions of values, while our datasets contain only stratified sampled records, making such distribution signals unreliable.
\name~ again performs the best among all competitors.

\subsection{Effect of Iterations}
\label{sec:exp_iteration}


\begin{figure}[t]
    \begin{minipage}[h]{0.49\linewidth}
        \centering
        \includegraphics[width=\linewidth]{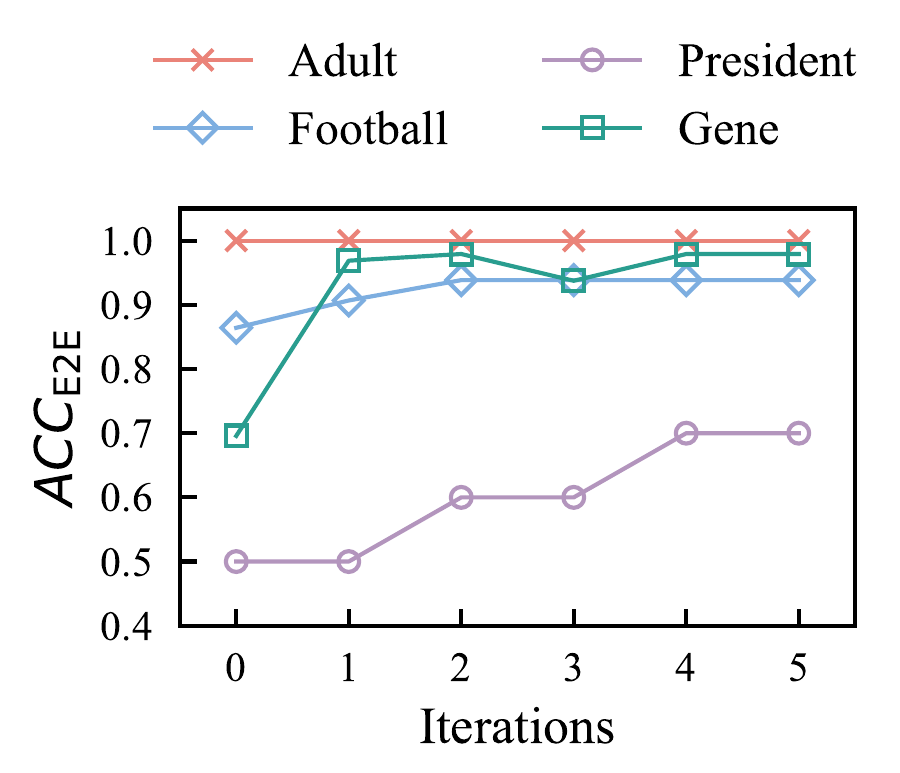}
        \vspace{-6mm}
        \subcaption{$Acc_\text{E2E}$}
    \end{minipage}
    \begin{minipage}[h]{0.49\linewidth}
        \centering
        \includegraphics[width=\linewidth]{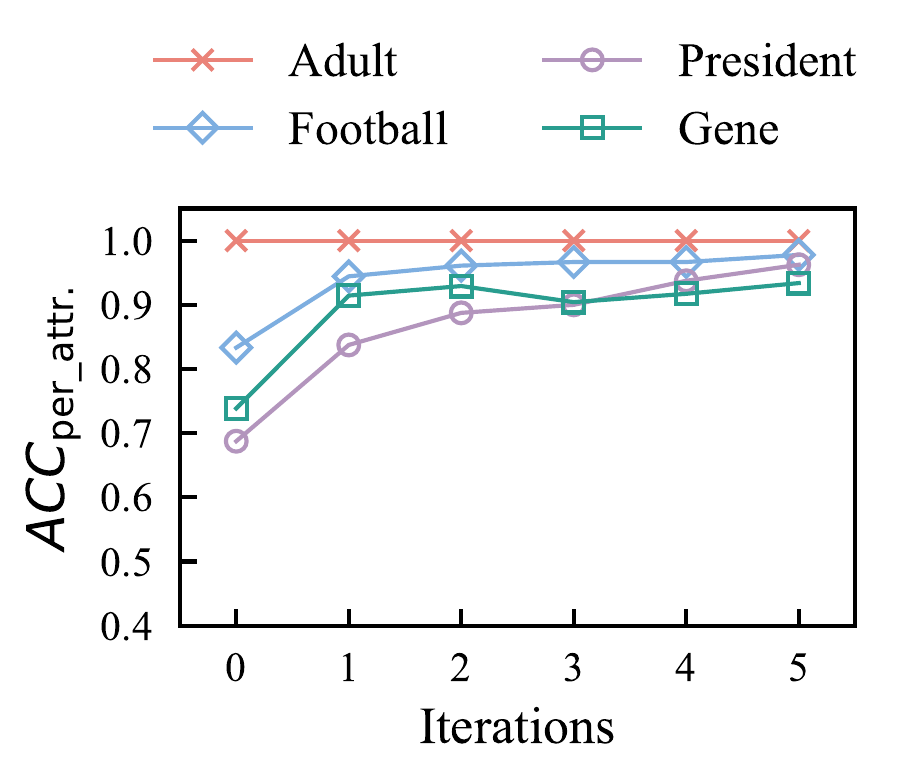}
        \vspace{-6mm}
        \subcaption{$Acc_\text{per\_attr.}$}
    \end{minipage} \\
    \vspace{-3mm}
    \caption{Performance of different iteration times}
    \label{fig:iterations}
    \vspace{-5mm}
\end{figure}

In this section, we investigate the impact of iteration times in \name~. We include iteration 0 as a baseline case where no iterative refinement is performed. In this setting, the pipeline degenerates into the NaiveP baseline, with the two components executing independently without interaction. We explore the accuracy with iterations from 0 to 5, and the results are plotted in Figure~\ref{fig:iterations}. For the Adult dataset, the accuracy remains $100\%$ from iteration 0 to iteration 5 since this dataset does not have many challenges in unpivot identification; for the other three datasets, the plot shows a significant increase in accuracy for iterations from 0 to 2; however, for iterations from 3 to 5, this increase diminishes quickly. The low $Acc_\text{E2E}$ and $Acc_\text{per\_attr.}$ observed at iteration 0 validate the \textbf{Challenge II} discussed in Section~\ref{sec:intro}, demonstrating that treating unpivot identification and schema matching as isolated processes yields suboptimal performance. The noticeable increase in accuracy across iterations from 0 to 2 indicates the effectiveness of Self-Refine for the unpivot attribute set identification task, but the markedly slowed and even stabilized improvement after iteration 3 shows that Self-Refine's capability has an upper bound; the LLM cannot infinitely improve its answers. As additional iterations incur higher computational and interaction costs (especially with large models) while offering only marginal accuracy gains, we set the number of iterations in \name~ to $2$ by default to balance the efficiency and accuracy. Two iterations already provide sufficient interaction between the two components to achieve a satisfying performance. 

\subsection{Effect of Stochastic Probability $\varepsilon$}
\label{sec:exp_epsilon}

In this section, we evaluate the effect of different probability $\varepsilon$ for random radius-$1$ modification in the expansion phase. We explore the accuracy for $\varepsilon \in \left\{ 0.05, 0.2, 0.4, 0.6, 0.8, 1.0\right\}$, and the results are plotted in Figure~\ref{fig:epsilon}. Since the Adult dataset is not challenging, and the initialization phase at iteration 0 can already reach $100\%$ accuracy (as verified in Section~\ref{sec:exp_iteration}), we focus on the performance on the other three datasets, i.e., Football, President and Gene. Figure~\ref{fig:epsilon} presents the different iteration-accuracy curves under different $\varepsilon$ values. Overall, the accuracy increases as the iterations increase from a broad perspective for all settings. When comparing across different $\varepsilon$ values, smaller $\varepsilon$ leads to a faster and more stable accuracy improvement, whereas larger $\varepsilon$ results in slower convergence and more frequent local fluctuations. This effect is more pronounced on the large-scale Gene dataset, where we observe more outliers in the curves, because its much larger search space makes the $\varepsilon$-induced random exploration markedly less effective under a limited number of iterations. This demonstrates that LLM-guided refinement effectively steers the search toward correct unpivot operators and accelerates convergence.

\begin{figure}[t]
    \begin{minipage}[h]{0.49\linewidth}
        \centering
        \includegraphics[width=\linewidth]{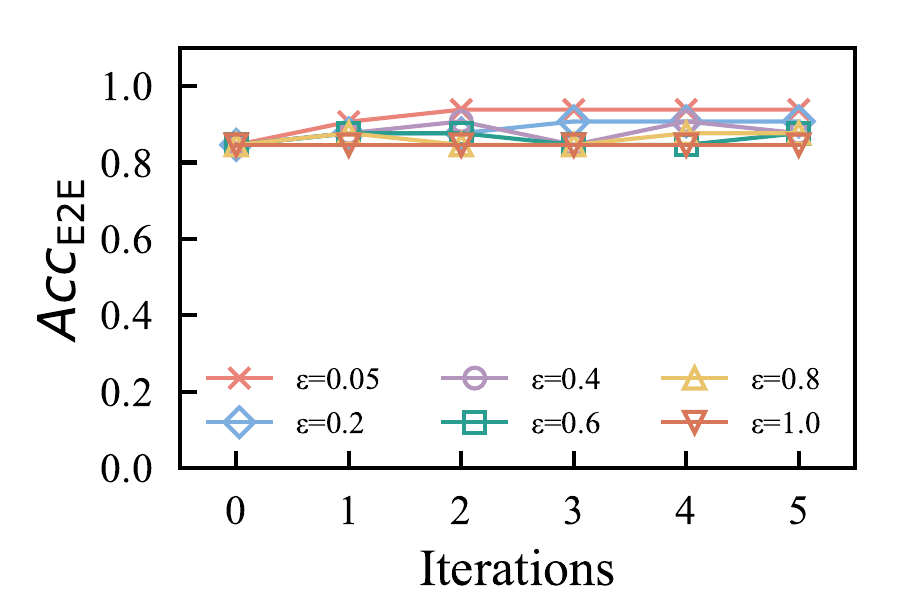}
        \vspace{-6mm}
        \subcaption{$Acc_\text{E2E}$ on Football}
    \end{minipage}
    \begin{minipage}[h]{0.49\linewidth}
        \centering
        \includegraphics[width=\linewidth]{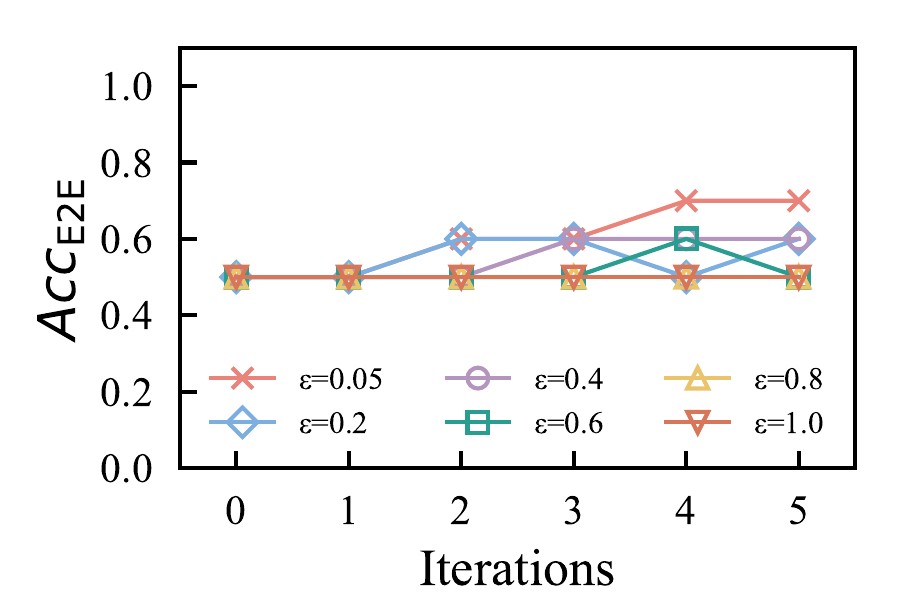}
        \vspace{-6mm}
        \subcaption{$Acc_\text{E2E}$ on President}
    \end{minipage} \\
    \begin{minipage}[h]{0.49\linewidth}
        \centering
        \includegraphics[width=\linewidth]{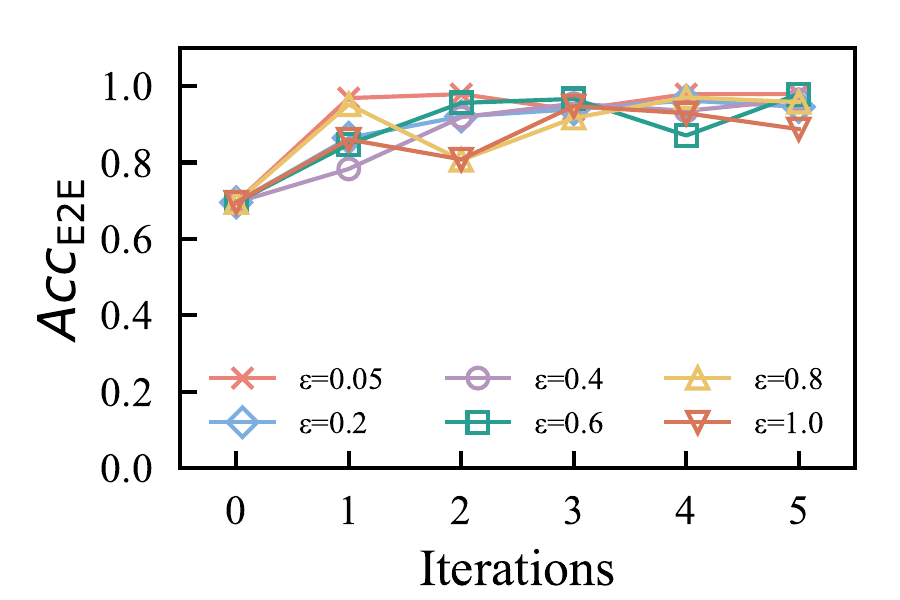}
        \vspace{-6mm}
        \subcaption{$Acc_\text{E2E}$ on Gene}
    \end{minipage}
    \begin{minipage}[h]{0.49\linewidth}
        \centering
        \includegraphics[width=\linewidth]{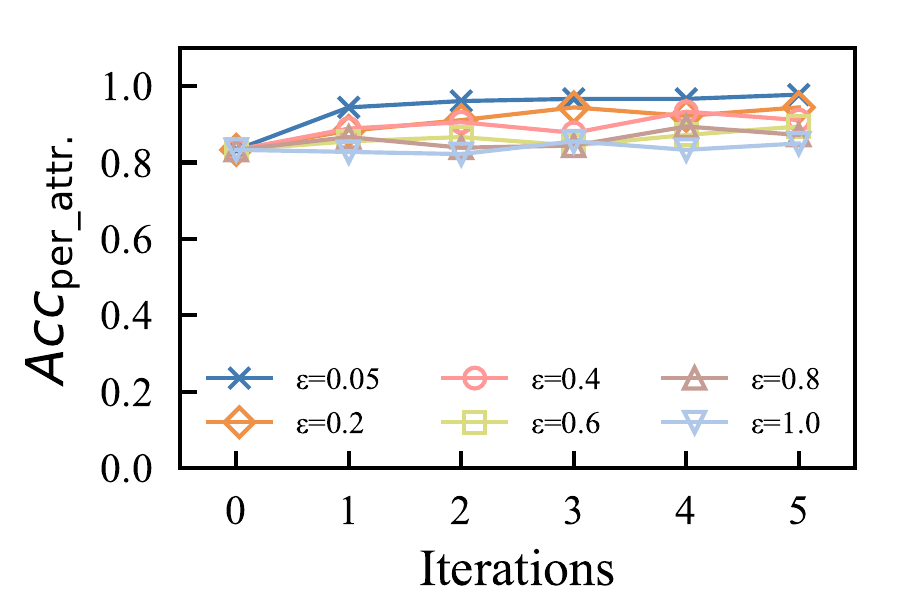}
        \vspace{-6mm}
        \subcaption{$Acc_\text{per\_attr.}$ on Football}
    \end{minipage} \\
    \begin{minipage}[h]{0.49\linewidth}
        \centering
        \includegraphics[width=\linewidth]{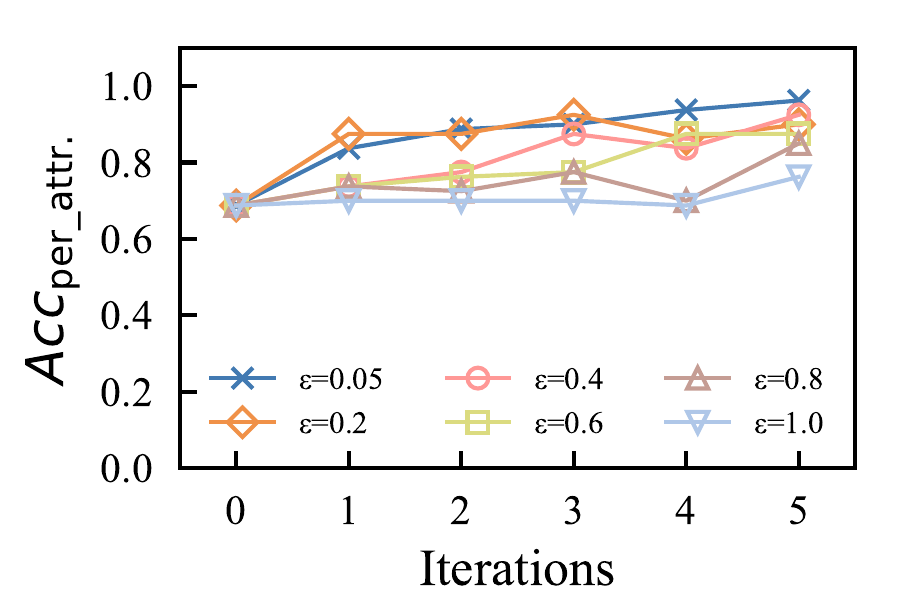}
        \vspace{-6mm}
        \subcaption{$Acc_\text{per\_attr.}$ on President}
    \end{minipage}
    \begin{minipage}[h]{0.49\linewidth}
        \centering
        \includegraphics[width=\linewidth]{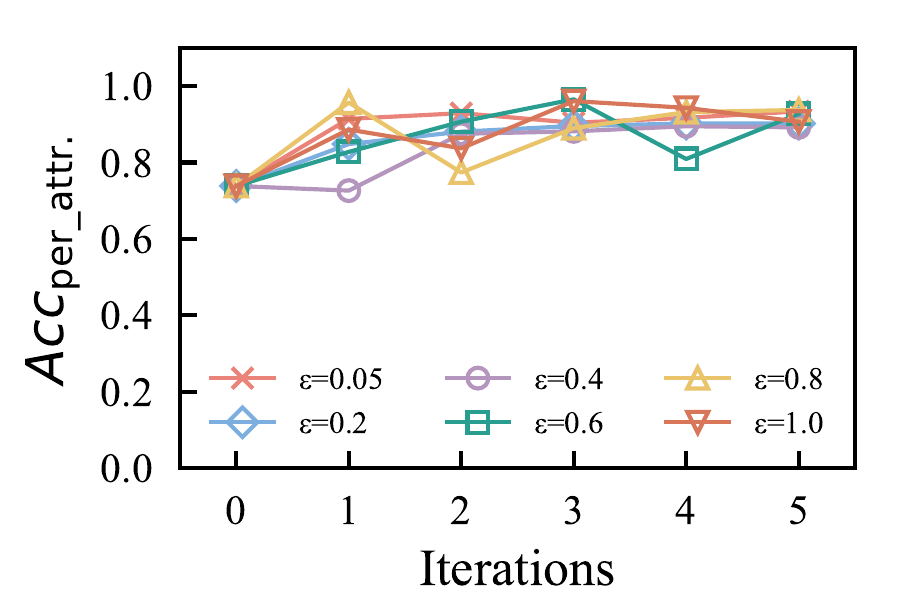}
        \vspace{-6mm}
        \subcaption{$Acc_\text{per\_attr.}$ on Gene}
    \end{minipage}
    \vspace{-3mm}
    \caption{Performance of different epsilon}
    \label{fig:epsilon}
    \vspace{-2mm}
\end{figure}

\subsection{Ablation Study on Root Node Generation}
In this section, we conduct an ablation study on the generation of the root node. We compared the accuracy of different generation methods of the unpivot attribute set for the root node or the Monte-Carlo search tree on Qwen3 models of different sizes. Our study shows that querying LLM for an initial set outperforms both simply selecting all the source attributes and selecting random attributes on the $235$B version Qwen3 model. More details of the evaluation results are presented in Appendix~\ref{app:ablation_study}.

\subsection{Scalability of \name~}

\begin{figure}[t]
    \centering
    \includegraphics[width=\linewidth]{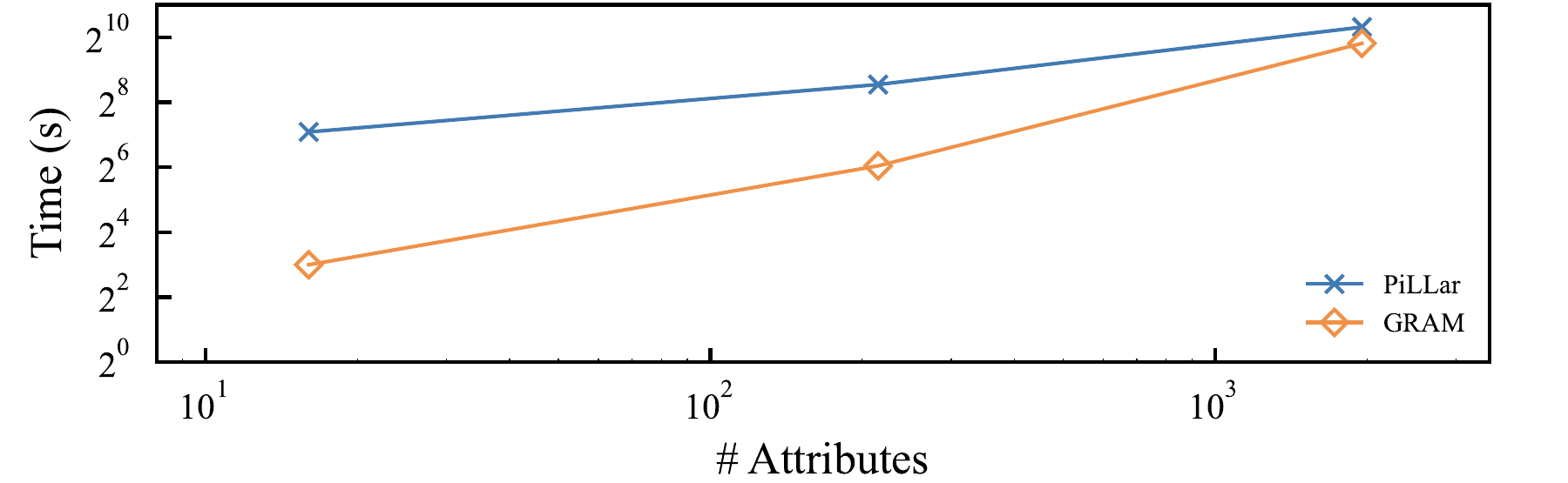}
    \vspace{-7mm}
    \caption{Runtime scalability w.r.t. the number of attributes}
    \label{fig:runtime_scalability}
    \vspace{-5mm}
\end{figure}

In this section, we evaluate the runtime scalability of \name~ on different numbers of attributes across three orders of magnitude ($10^1$, $10^2$ and $10^3$). To evaluate the $10^3$-scale runtime, we additionally use the M5 Forecasting (Walmart sales) dataset from Kaggle~\footnote{\url{https://www.kaggle.com/competitions/m5-forecasting-accuracy}}, which contains around $2000$ time-series columns.
As mentioned in \textbf{Challenge II} discussed in Section~\ref{sec:intro}, matching for pivot table schema inherently requires LLM reasoning. Therefore, we compare runtime scalability against a representative LLM-based baseline, GRAM.
As shown in Figure~\ref{fig:runtime_scalability}, \name~ exhibits a noticeably flatter growth trend than GRAM as the number of attributes increases, indicating better scalability for large schemas. This is because GRAM issues one LLM call per attribute, so its total number of calls grows linearly with the number of attributes; with the increase in per-call latency as the prompt becomes longer on larger schemas, the runtime grows worse than linear in practice. In contrast, our method keeps the number of LLM calls constant, so the overhead comes only from the modest increase in prompt length within a fixed number of calls, yielding a much more gradual runtime increase as the schema scales. As discussed in Section~\ref{sec:realization}, we consider this runtime as acceptable given the significant performance gains.
\section{Related Work}
\label{sec:related_work}

\noindent \textbf{\textit{Rule-based approaches.}} Rule-based schema matching relies on manually designed rules to measure similarity between source and target attributes. Representative systems~\cite{palopoli2000system, madhavan2001generic, doan2000learning, zhang2025smutf} combine lexical, structural, and distributional heuristics, usually with fixed or lightweight aggregation. COMA~\cite{do2002coma} integrates multiple handcrafted matchers via predefined combine strategies (e.g., average/max/min-threshold). DistributionBased~\cite{zhang2011automatic} leverages distribution similarity and intersection signals to group semantically related columns into shared attributes. While these methods are efficient and interpretable, they often miss deeper semantic correspondences, and designing robust rules remains non-trivial.

\noindent \textbf{\textit{Deep-learning-based approaches.}} Deep-learning methods (semi-)automatically discover schema matches using neural models. Many works~\cite{zhang2021smat, zhang2023schema, wu2023conschema, tu2023unicorn} encode attribute names (and optionally metadata) with pretrained language models, then compute similarities via neural modules or self-training. ADnEV~\cite{shraga2020adnev} instead refines a similarity matrix produced by conventional matchers using two neural networks. Compared to rule-based methods, these approaches better capture semantics and can adapt with labeled data or feedback, but they typically require substantial training data; models trained on one domain may also generalize poorly to new domains.

\noindent \textbf{\textit{LLM-based approaches.}} Recent work leverages LLMs for schema matching. GRAM~\cite{liu2024gram} combines named-entity-resolution (NER) and retrieval-augmented generation (RAG)~\cite{lewis2020retrieval} to generate privacy-aware prompts and improve efficiency and accuracy. Other systems~\cite{seedat2024matchmaker, liu2025magneto, parciak2150schema} use LLMs for reranking, synthetic in-context learning, or hybrid pipelines, largely treating LLMs as semantic reasoners or validators rather than redesigning the workflow. These approaches offer strong semantic generalization and can work under privacy constraints, but still struggle with ambiguous domain abbreviations and cases requiring knowledge beyond semantics.

\section{Conclusions}
\label{sec:conclusions}

In this paper, we study the joint schema-value matching problem between pivot tables and standard relational tables under the setting where only a minimum of data records can be accessed due to privacy concerns. We present \name~, an LLM-driven matching for pivot table schema method that relies on our proposed MCTS-based search paradigm. In \name~, we divide the matching generation process into two stages, namely schema flattening and schema matching, which are iteratively executed and mutually adjust each other. Schema flattening is a component that identifies the unpivot operator for the input pivot table, and schema matching is a component that generates matches and evaluates the reward. Our experiments shows the superiority of \name~. As for future works, a promising direction is to handle more complex mapping transformations for tables in the wild.

\bibliographystyle{ACM-Reference-Format}
\bibliography{refer}

\appendix

\section{Supplementary Case Study Result}
\label{app:case_study}

This appendix provides supplementary qualitative outputs for the running example in Figure~\ref{fig:llm_responses}, together with the web-chat prompt transcript used to obtain them (we omit intermediate assistant acknowledgements for brevity.).

\small
\begin{framed}
\noindent [Turn 1] user: \newline
\noindent You are now an expert in data governance, first I'll give you a definition, a requirement and some examples, and I need you to remember them for the following request. \newline

\noindent [Turn 2] user: \newline
\noindent Definition: \newline
\noindent Unpivot: Transforming multiple horizontally arranged numeric columns into vertical attribute-value pairs, preserving identifier columns, where original column names become values in a new attribute column and their corresponding data is consolidated into a unified value column. \newline
\noindent Requirement: \newline
\noindent Your task is to detect the attributes that can be unpivoted in the source table. A source table and a target table for reference will be provided. Your answer should be in JSON format and no explanation is needed. For example, if the attributes to be unpivoted is $\left[ \text{A}, \text{B}, \text{C}, \text{D} \right]$, your answer should be $\{ ``\text{unpivot}\_\text{columns}": \left[ ``\text{A}", ``\text{B}", ``\text{C}", ``\text{D}" \right] \}$. If no attribute is in the unpivot subset, answer with an empty unpivot\_columns array, that is, $\{ ``\text{unpivot}\_\text{columns}": \left[ \right] \}$. And remember that your JSON string should be pure text, do not put it in a code block. \newline
\noindent Example: \newline
\noindent For input attributes $\left[ \text{Product}, \text{Jan}\_\text{Sales}, \text{Feb}\_\text{Sales} \right]$, the corresponding output attribute is $\left[ \text{Jan}\_\text{Sales}, \text{Feb}\_\text{Sales} \right]$, and the output answer should be $\{``\text{unpivot}\_\text{columns}": \left[ ``\text{Jan}\_\text{Sales}", ``\text{Feb}\_\text{Sales}" \right] \}$. \newline
\noindent Example: \newline
\noindent For input attributes $\left[ \text{Trade}, \text{Date}, \text{Quantity} \right]$, the corresponding output attribute is $\left[ \right]$, because there is no attribute to be unpivoted, and the output answer should be $\{``\text{unpivot}\_\text{columns}": \left[ \right] \}$. \newline

\noindent [Turn 3] user: \newline
\noindent \#\#\# Identify the columns that can be unpivoted in a list of column names and with no explanation. \newline
\noindent \#\#\# Source column names: \newline
\noindent \# \newline
\noindent \# Div \newline
\noindent \# Date \newline
\noindent \# HS \newline
\noindent \# AS \newline
\noindent \# HST \newline
\noindent \# AST \newline
\noindent \# \newline
\noindent \#\#\# Description: \newline
\noindent \# \newline
\noindent \# Div: League division abbreviation (e.g., `E0' for English Premier League) \newline
\noindent \# Date: Match date (format: DD/MM/YY) \newline
\noindent \# HS: Home Shots (total shots attempted by the home team) \newline
\noindent \# AS: Away Shots (total shots attempted by the away team) \newline
\noindent \# HST: Home Shots on Target (shots on goal by the home team) \newline
\noindent \# AST: Away Shots on Target (shots on goal by the away team) \newline
\noindent \# \newline
\noindent \#\#\# Target column names for reference: \newline
\noindent \# \newline
\noindent \# Div \newline
\noindent \# Date \newline
\noindent \# Metric \newline
\noindent \# Value
\end{framed}
\normalsize

\begin{figure}[t]
    \centering
    \includegraphics[width=\linewidth]{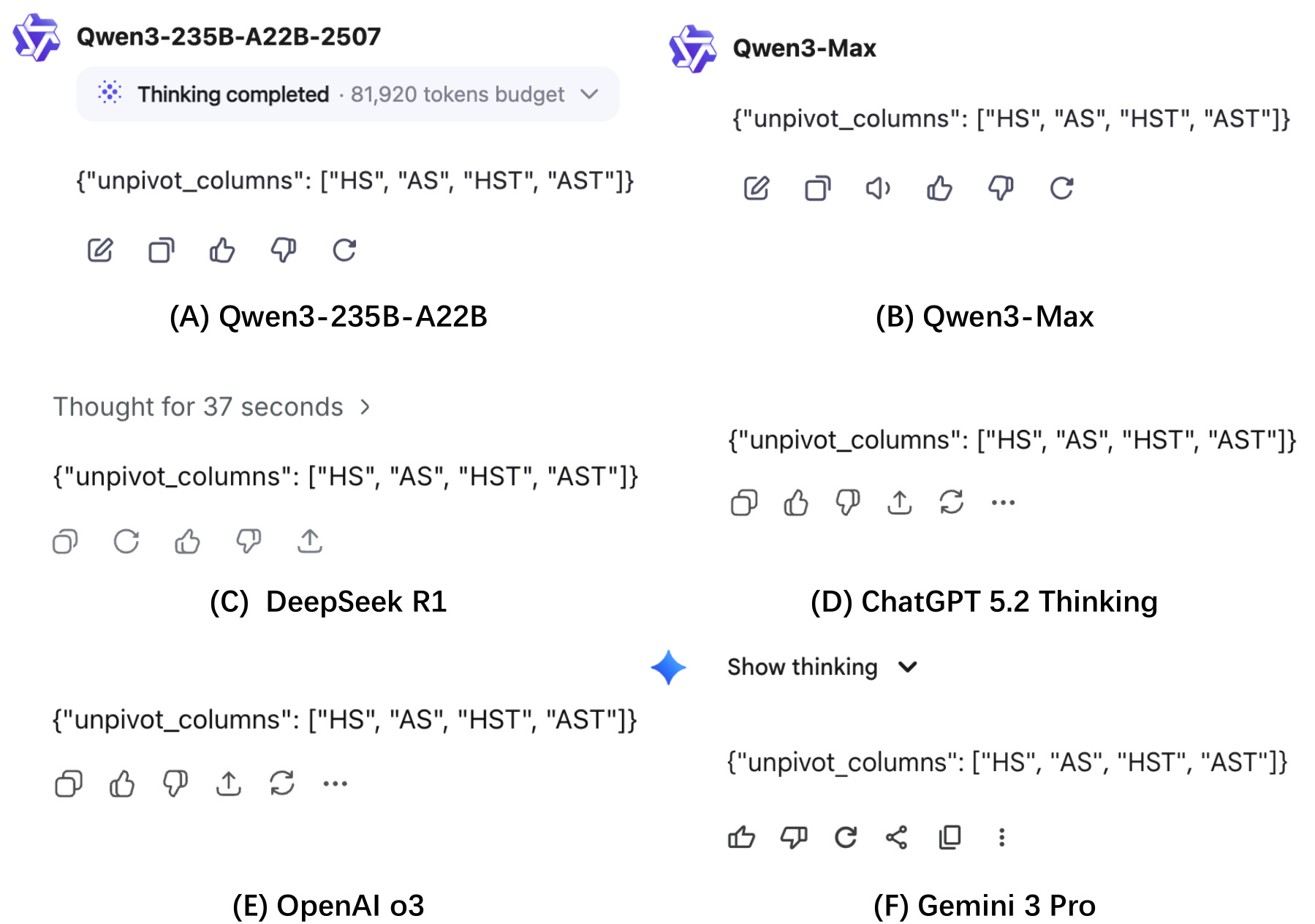}
    \vspace{-4mm}
    \caption{Unpivoted attributes identified by SOTA LLMs}
    \label{fig:llm_responses}
    \vspace{-4mm}
\end{figure}
\section{Detailed Proofs for Convergence Analysis}
\label{app:proofs}


\begin{assumption}[$\varepsilon$-Randomized Expansion]\label{asp:exp}
During expansion, a new child is generated either by an LLM-guided refinement with probability $1-\varepsilon$, or by a radius-$1$ random modification with probability $\varepsilon>0$. The random modification assigns a strictly positive probability to every yet-unexpanded neighbor candidate of the current node, and duplicate generations are forbidden. Consequently, every feasible candidate in the finite search space will eventually be generated with reachability~$1$.
\end{assumption}

\begin{lemma}[Exploration Completeness]\label{lem:exploration-complete}
Under Assumption~\ref{asp:exp}, the randomized expansion mechanism is probabilistically complete: every feasible candidate in the finite search space will eventually be generated with probability~$1$.
\end{lemma}

\begin{proof}[Sketch]
Consider any candidate $u$ that has not yet been generated. Whenever its parent node is expanded, the random expansion branch is taken with probability $\varepsilon>0$, and conditional on that branch, $u$ is selected with some fixed probability $p_u>0$. Since duplicate generations are forbidden, the probability that $u$ is never generated after $K$ such expansion trials is at most $(1-\varepsilon p_u)^K$, which converges to $0$ as $K\to\infty$. Because the candidate space is finite, applying this argument to all remaining candidates ensures that each will be generated in finite time with probability~$1$.
\end{proof}

\begin{assumption}[UCT Selection: Infinite Visits]\label{asp:uct}
Internal selection follows the UCT rule. As visit counts grow, every edge (in particular, the optimal child on each prefix of the witnessing path) is visited infinitely many times as $t\to\infty$.
\end{assumption}


\noindent \textbf{Strong vs.\ Weak Hits.}
A \emph{strong hit} occurs when the optimal node $v^\star$ is first generated and its reward $R(v^\star)$ is backpropagated, or when a selection reaches the already-generated $v^\star$ and backpropagates $R(v^\star)$. A \emph{weak hit} denotes a non-optimal expansion.
Let the $j$-th \emph{block} be the interval between two consecutive strong hits $\sigma_{j - 1}\!\to\!\sigma_j$ and define the root endpoint difference
\begin{equation}
W_j := Q^{(\mathrm{end\ of\ block}\ j)}_{v_0} - Q^{(\mathrm{start\ of\ block}\ j)}_{v_0},\qquad |W_j|\le \Omega.
\end{equation}
Let $E_j := |R(s^\star)-Q^{(\mathrm{after}\ \sigma_j)}_{v_0}|$ be the root error after the $j$-th strong hit.

\noindent \textbf{Single Strong Hit Kernel.}
Writing the one-step (single strong hit) recursion layerwise as in Equation~\eqref{eq:maxavg} gives, for $d=0, \ldots, H - 2$,
\begin{equation}
e'_d \le \tfrac12 e_d + \tfrac12 e'_{d + 1},\qquad e'_H=0,
\end{equation}
which compactly yields
\begin{equation}\label{eq:single-kernel}
e' \leq \tfrac12 \big( I- \tfrac12 S \big)^{-1} e
= \sum_{k = 0}^{H} \frac{1}{2^{k + 1}} S^k e.
\end{equation}
In particular, for the root component
\begin{equation}\label{eq:cH}
\big| e'_0 \big| \leq \sum_{k = 0}^{H} \frac{1}{2^{k+1}} \lvert e_k \rvert
\leq \Big( 1 - \frac{1}{2^{H + 1}} \Big) \, \| e \|_\infty
=: c_H \|e\|_\infty, \quad c_H \in (0,1).
\end{equation}

\noindent \textbf{Block Recursion.}
Combining the single-hit contraction at the block boundary with the intra-block drift $W_j$, we obtain
\begin{equation}\label{eq:block-recursion}
E_j \leq c_H E_{j-1} + c_H |W_j|, \qquad c_H = 1 - 2^{-(H + 1)}.
\end{equation}

\noindent \textbf{ISS Baseline.}
Iterating Equation~\eqref{eq:block-recursion} gives
\begin{equation}\label{eq:ISS}
E_m \leq c_H^{ \, m} E_0 + c_H \sum_{r = 1}^{m} c_H^{\,m - r}|W_r|,
\end{equation}
i.e., geometric stability modulo the disturbance sequence $\{|W_r|\}$.

\noindent \textbf{Why $\boldsymbol{|W_j|\!\to\!0}$.}

\emph{(i) Bottom-up propagation starts at $v^\star$.}
Once the optimal leaf $v^\star$ is discovered, its value fixes at $Q_{v_H}=R(v^\star)$.
When its parent $u=v_{H - 1}$ is updated, the max–average backup
\begin{equation}
Q_u \leftarrow \tfrac12 \!\Big( Q_u + \max\{\, R(u), \, Q_{v^\star} \, \} \Big)
\end{equation}
contracts $Q_u$ toward $\max \{R(u), R(v^\star) \} = R(v^\star)$. Hence repeated revisits \emph{along the optimal child} drive $Q_u \! \to \! R(v^\star)$; once $Q_u$ is close enough to $R(v^\star)$, UCT at $v_{H - 2}$ increasingly favors $u$ as its optimal child, and the same argument repeats upward.

\emph{(ii) UCT makes shallow weak selections asymptotically negligible.}
Under deterministic rewards, UCT at any internal node $v_d$ asymptotically favors its optimal child. Consequently, for each fixed depth $d$, the empirical ratio of selecting any strictly suboptimal child tends to zero as visits grow. On the witnessing path $(v_0, \ldots, v_{H - 1})$, this implies that updates using the optimal child dominate in the limit at every shallow prefix, while weak selections occur only finitely many times or with vanishing frequency.

\emph{(iii) Vanishing block-end drift.}
Within block $j$, let $W_j := Q_{v_0}^{\mathrm{end}(j)} - Q_{v_0}^{\mathrm{start}(j)}$ denote the raw root drift accumulated by weak updates before the block-ending strong hit. The subsequent strong hit applies the single-hit contraction to \emph{both} the inherited error and this accumulated drift, so its contribution to the post-block error is exactly $c_H|W_j|$ in the recursion
\begin{equation}
E_j \; \leq \; c_H E_{j - 1} \; + \; c_H|W_j|.
\end{equation}
By (i) and (ii), UCT makes shallow weak selections asymptotically negligible while bottom-up propagation repeatedly pulls ancestors toward $R(v^\star)$; hence the raw drift $|W_j| \to 0$. 
Therefore the contracted disturbance $c_H|W_j|$ also vanishes, yielding $E_j \to 0$.

Remark that because $p^\star$ ignores the children generated by LLM Self-Refine, it is conservative; in practice, convergence is typically much faster.
\section{Prompt Templates}

\subsection{Initialization Prompt}

\small
\begin{framed}
\noindent \textbf{Role:} system \newline
\noindent \textbf{Content:} \newline
\noindent You are now an expert in data governance, first I'll give you a definition, a requirement and some examples, and I need you to remember them for the following request. \newline

\noindent \textbf{Role:} user \newline
\noindent \textbf{Content:} \newline
\noindent Definition: \newline
\noindent Unpivot: Transforming multiple horizontally arranged numeric columns into vertical attribute-value pairs, preserving identifier columns, where original column names become values in a new attribute column and their corresponding data is consolidated into a unified value column. \newline
\noindent Requirement: \newline
\noindent Your task is to detect the attributes that can be unpivoted in the source table. A source table and a target table for reference will be provided. Your answer should be in JSON format and no explanation is needed. For example, if the attributes to be unpivoted is $\left[ \text{A}, \text{B}, \text{C}, \text{D} \right]$, your answer should be $\{ ``\text{unpivot}\_\text{columns}": \left[ ``\text{A}", ``\text{B}", ``\text{C}", ``\text{D}" \right] \}$. If no attribute is in the unpivot subset, answer with an empty unpivot\_columns array, that is, $\{ ``\text{unpivot}\_\text{columns}": \left[ \right] \}$. And remember that your JSON string should be pure text, do not put it in a code block. \newline
\noindent Example: \newline
\noindent For input attributes $\left[ \text{Product}, \text{Jan}\_\text{Sales}, \text{Feb}\_\text{Sales} \right]$, the corresponding output attribute is $\left[ \text{Jan}\_\text{Sales}, \text{Feb}\_\text{Sales} \right]$, and the output answer should be $\{``\text{unpivot}\_\text{columns}": \left[ ``\text{Jan}\_\text{Sales}", ``\text{Feb}\_\text{Sales}" \right] \}$. \newline
\noindent Example: \newline
\noindent For input attributes $\left[ \text{Trade}, \text{Date}, \text{Quantity} \right]$, the corresponding output attribute is $\left[ \right]$, because there is no attribute to be unpivoted, and the output answer should be $\{``\text{unpivot}\_\text{columns}": \left[ \right] \}$. \newline

\noindent \textbf{Role:} assistant \newline
\noindent \textbf{Content:} \newline
\noindent Got it! Please provide the source and target tables so I can determine the unpivot columns and provide the JSON output. \newline

\noindent \textbf{Role:} user \newline
\noindent \textbf{Content:} \newline
\noindent \#\#\# Identify the columns that can be unpivoted in a list of column names and with no explanation. \newline
    \noindent \#\#\# Source column names: \newline
    \noindent \# \newline
    \noindent \# $\ldots$ \textcolor[rgb]{0.7,0.7,0.7}{\textit{source attributes}} \newline
    \noindent \# \newline
    \noindent \#\#\# Description: \newline
    \noindent \# \newline
    \noindent \# $\ldots$ \textcolor[rgb]{0.7,0.7,0.7}{\textit{attribute descriptions}} \newline
    \noindent \# \newline
    \noindent \#\#\# Target column names for reference: \newline
    \noindent \# \newline
    \noindent \# $\ldots$ \textcolor[rgb]{0.7,0.7,0.7}{\textit{target attributes}} \newline
\end{framed}
\normalsize

\subsection{Feedback Prompt}

\small
\begin{framed}
\noindent \textbf{Role:} system \newline
\noindent \textbf{Content:} \newline
\noindent You are now an expert in data governance and schema matching, and provides feedback on the quality of unpivot detection. \newline

\noindent \textbf{Role:} user \newline
\noindent \textbf{Content:} \newline
\noindent \#\#\# Evaluate the columns selected to be unpivoted from the source table. The selection aims to transfer the source table to the target table. You should focus on the transformation between the source and target table structure rather than the meaning of unpivot. The provided sample data may have been anonymised. Analyze this answer strictly and critically, point out every flaw for every possible imperfection about the selection. You only need to evaluate the selection of unpivot subset itself. Note that the selected subset is under loose limits, your task is to reduce the size of the subset if there exists redundant attributes in the subset. Remember the attributes should be selected from the source attributes, do not use names that do not exist. \newline
\noindent \#\#\# Source column names:\newline
\noindent \# \newline
\noindent \# $\ldots$ \textcolor[rgb]{0.7,0.7,0.7}{\textit{source attributes}} \newline
\noindent \# \newline
\noindent \#\#\# Description: \newline
\noindent \# \newline
\noindent \# $\ldots$ \textcolor[rgb]{0.7,0.7,0.7}{\textit{attribute descriptions}} \newline
\noindent \# \newline
\noindent \#\#\# Target column names for reference: \newline
\noindent \# \newline
\noindent \# $\ldots$ \textcolor[rgb]{0.7,0.7,0.7}{\textit{target attributes}} \newline
\noindent \# \newline
\noindent \#\#\# Sample data from source table: \newline
\noindent \# \newline
\noindent \# $\ldots$ \textcolor[rgb]{0.7,0.7,0.7}{\textit{sample data from source table}} \newline
\noindent \# \newline
\noindent \#\#\# Sample data from target table: \newline
\noindent \# \newline
\noindent \# $\ldots$ \textcolor[rgb]{0.7,0.7,0.7}{\textit{sample data from target table}} \newline
\noindent \# \newline
\noindent \#\#\# Selected columns for unpivot: \newline
\noindent \# \newline
\noindent \# $\left[ \ldots \textcolor[rgb]{0.7,0.7,0.7}{\textit{attributes}} \right]$ \newline
\end{framed}
\normalsize

\subsection{Refine Prompt}

\small
\begin{framed}
\noindent \textbf{Role:} user \newline
\noindent \textbf{Content:} \newline
\noindent \#\#\# Refine your selection based on the feedback. If the feedback indicates that the selection is ideal, then you can remain the selection unchanged. Note that the suggested subset provided in the feedback may contain attributes that are not in the source table, you should not totally rely on it, but rather use it as a reference and strictly select from source attributes. \newline
\noindent \#\#\# Feedback: \newline
\noindent \# \newline
\noindent \# $\ldots$ \textcolor[rgb]{0.7,0.7,0.7}{\textit{feedback}} \newline
\end{framed}
\normalsize
\section{More Experiment Details}
\label{app:exp_details}

\subsection{Details of Baselines}
\label{app:baseline_details}

We compare \name~ against the following representative approaches:

\begin{itemize}[leftmargin=*]
\label{exp:baseline}
\item \textbf{COMA 3.0}~\cite{massmann2011evolution}~\footnote{We use the implementation available from the Valentine package~\cite{koutras2021valentine}. \label{footnote:valentine}}. COMA 3.0 is a multi-matcher schema matching framework that combines linguistic, structural, and instance-based matchers by an average-based similarity aggregation. It also employs advanced strategies such as fragment matching and filtered context to efficiently handle large-scale tasks.

\item \textbf{DisB}~\cite{zhang2011automatic}~\textsuperscript{\ref{footnote:valentine}}.
DisB automatically clusters attributes into semantically coherent attributes using purely data-driven evidence. It clusters relational attributes into semantically coherent attributes via distribution similarity, and then refines them via intersection-based similarity and witness columns using correlation clustering.

\item \textbf{GRAM}~\cite{liu2024gram}.
GRAM is an LLM-based schema matching framework that integrates retrieval augmentation and prompt compression to accelerate inference while maintaining accuracy. It employs a Named Entity Recognition (NER)~\cite{chinchor1997muc} filter and a Double-RAG~\cite{lewis2020retrieval} mechanism to dynamically select relevant target attributes and few-shot examples, forming a compact and adaptive prompting process for efficient attribute alignment.

\item \textbf{NaiveP}. Naive Pipeline (NaiveP) represents a straightforward execution pipeline. It performs unpivot attribute identification and schema matching independently without iterative refinement. In this setting, the LLM first generates an unpivot attribute set, and the corresponding tables are directly fed into the schema matching component without further adjustment.
\end{itemize}

For GRAM, since we are not able to get access to its source code, we implement the competitor according to the design and the prompt provided in the paper. Different from the original implementation, we use the Qwen3 model instead of the original FLAN-T5 model for a fair evaluation of the performance.

\subsection{Details of Implementation}
\label{app:implementation_details}

We detail the hyper-parameters used in \name~ as follows.
We adopt Qwen3 provided by the Aliyun Model Studio API as the LLM model. Embeddings of attribute names are generated by the fine-tuned DistillRoBERTa~\cite{sanh2019distilbert, liu2019roberta}~\footnote{We use the SentenceTransformers Python package with sentence-transformers/all-distilroberta-v1 model in code implementation.}.
The probability for the bounded stochastic policy is set to $\varepsilon = 0.05$. The constant $C$ for UCT selection is set to $2$, which is a commonly adopted choice in MCTS-based systems~\cite{hamrick2020combining, kohankhaki2024monte}. The maximum number of child nodes is set to $3$ and the executed iterations count is set to $2$. In each iteration, $5$ threads are submitted and execute the selection and expansion in parallel.
Unless explicitly specified, all hyper-parameters are set to their default values.
All experiments are executed on macOS Sequoia 15.6 with 8 physical CPU cores and 16GB of memory. The programs are all implemented in Python.
\section{Additional Experiments}
\label{app:additional_experiments}

\subsection{Ablation Study}
\label{app:ablation_study}

\noindent \textbf{Description Information.} We conduct an ablation study on the use of description information. This information contains descriptions of the attributes of the two input tables which help LLM understand the semantics of these attributes. In practical data governance settings, description information is typically readily available, as it is recommended as a standard component of well-maintained schemas~\cite{edara2021big, foundry_overview, powercenter_edit_column}. As shown in Figure~\ref{fig:description}, for large-scale models with $235$B parameters, removing the description information leads to a clear decrease in accuracy. This indicates that description information is a useful supplementary information for LLM in matching for pivot table schema tasks, especially under circumstances that attributes contain abbreviations and terminologies. Yet, this effect is not obvious for small models due to their weaker semantic understanding capacity, which is insufficient to fully exploit the fine-grained information contained in the descriptions. This is especially evident on the large-scale Gene dataset: incorporating description substantially increases the prompt length and instead reduces accuracy.

\begin{figure}[t]
    \begin{minipage}[h]{\linewidth}
        \centering
        \includegraphics[width=0.75\linewidth]{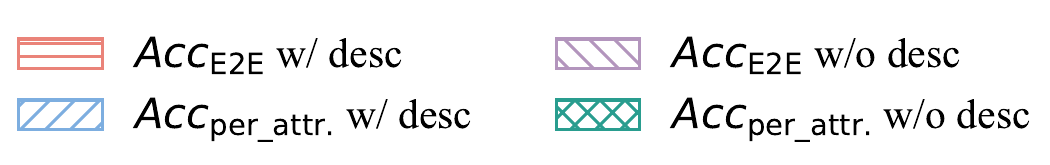}
    \end{minipage} \\
    \begin{minipage}[h]{0.49\linewidth}
        \centering
        \includegraphics[width=\linewidth]{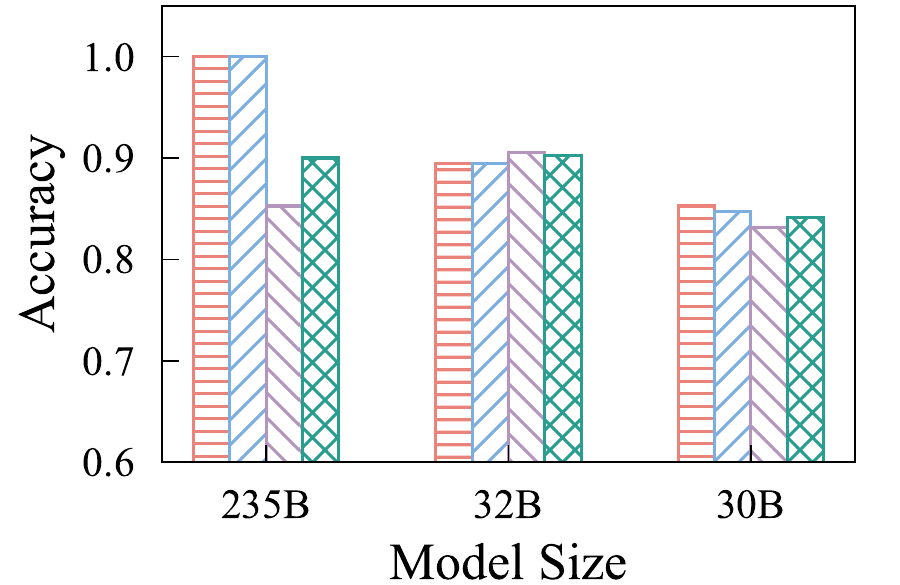}
        \subcaption{Adult}
    \end{minipage}
    \begin{minipage}[h]{0.49\linewidth}
        \centering
        \includegraphics[width=\linewidth]{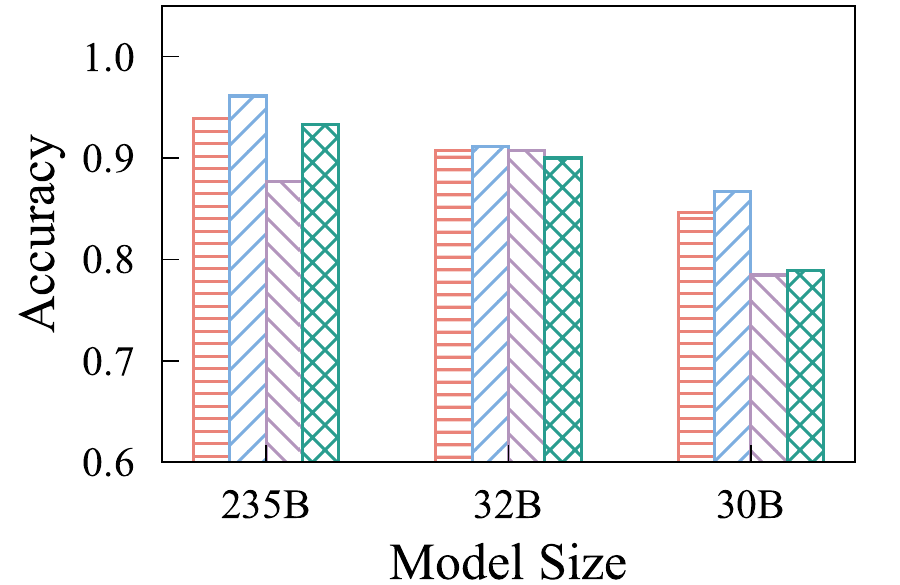}
        \subcaption{Football}
    \end{minipage} \\
    \begin{minipage}[h]{0.49\linewidth}
        \centering
        \includegraphics[width=\linewidth]{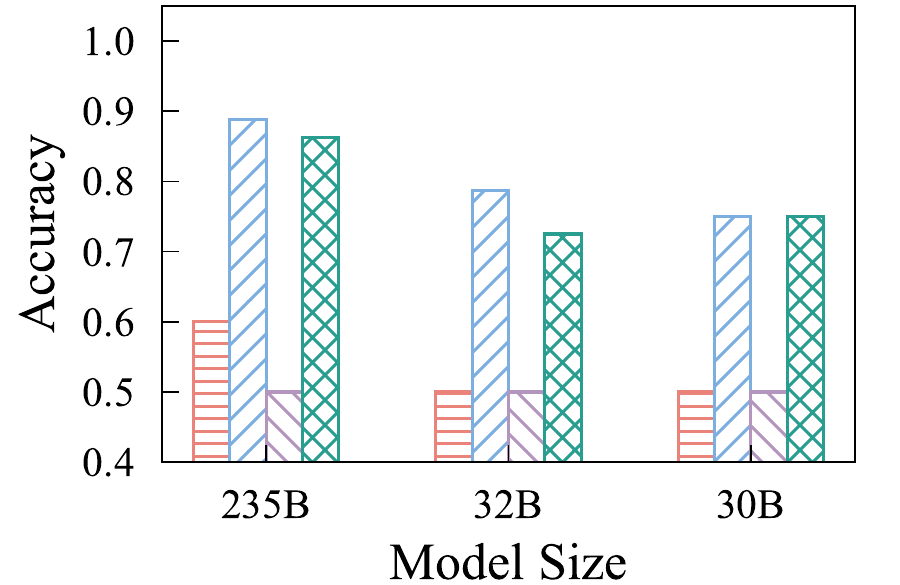}
        \subcaption{President}
    \end{minipage}
    \begin{minipage}[h]{0.49\linewidth}
        \centering
        \includegraphics[width=\linewidth]{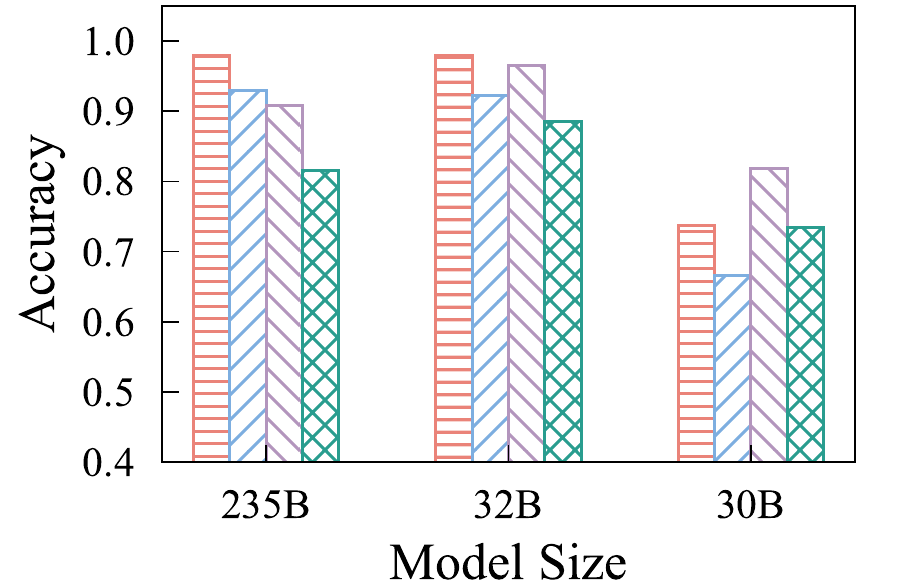}
        \subcaption{Gene}
    \end{minipage}
    \vspace{-3mm}
    \caption{Ablation study on description information}
    \label{fig:description}
    \vspace{-5mm}
\end{figure}

\begin{figure}[t]
    \begin{minipage}[h]{\linewidth}
        \centering
        \includegraphics[width=\linewidth]{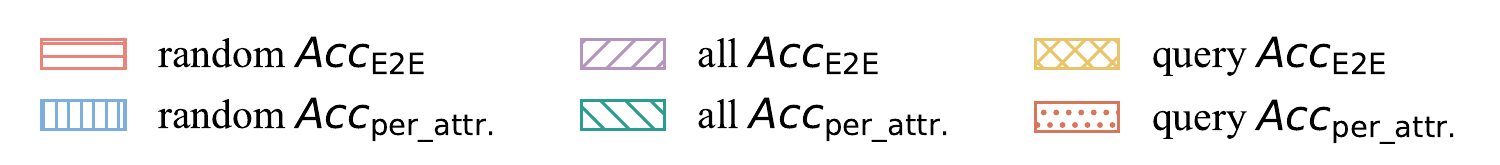}
    \end{minipage} \\
    \begin{minipage}[h]{0.49\linewidth}
        \centering
        \includegraphics[width=\linewidth]{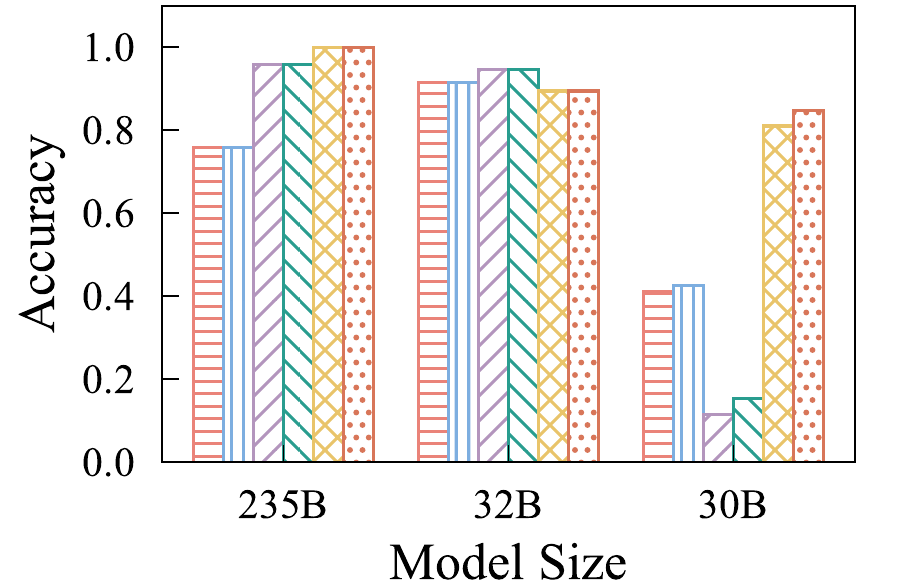}
        \subcaption{Adult}
    \end{minipage}
    \begin{minipage}[h]{0.49\linewidth}
        \centering
        \includegraphics[width=\linewidth]{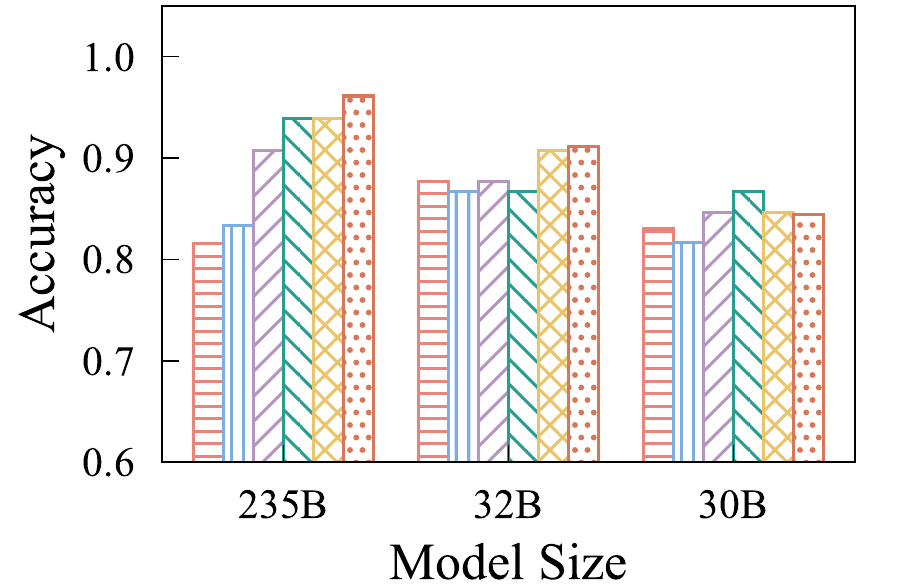}
        \subcaption{Football}
    \end{minipage} \\
    \begin{minipage}[h]{0.49\linewidth}
        \centering
        \includegraphics[width=\linewidth]{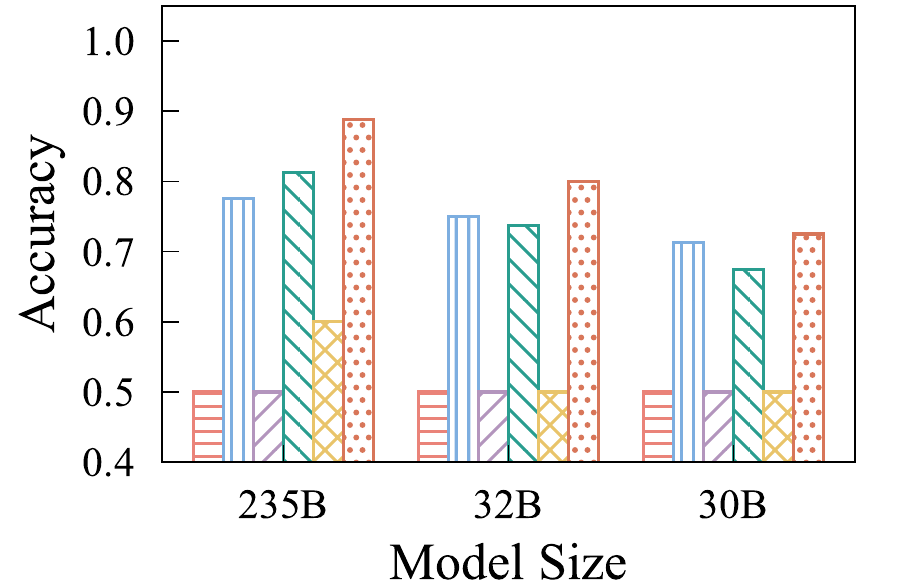}
        \subcaption{President}
    \end{minipage}
    \begin{minipage}[h]{0.49\linewidth}
        \centering
        \includegraphics[width=\linewidth]{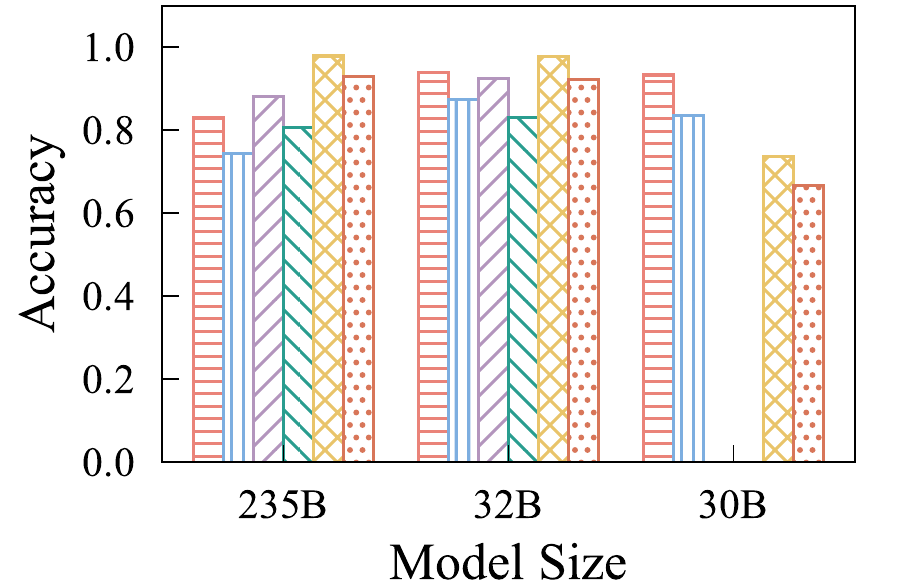}
        \subcaption{Gene}
    \end{minipage}
    \vspace{-3mm}
    \caption{Ablation study on the generation of root node}
    \label{fig:initialization}
    \vspace{-5mm}
\end{figure}

\noindent \textbf{Generation of the Root Node.} Figure~\ref{fig:initialization} plots the performance of different generation methods of the unpivot attribute set for the root node of the Monte-Carlo search tree. The method we use in \name~ is to query LLM and let it give an initial set. We compare it to two other methods: simply selecting all the source attributes and selecting random attributes. For the $235$B model, querying LLM shows notable superiority in accuracy on all datasets. This is because using the strong semantic capability of LLM, the initial set is quite close to the correct one, and subsequent refinement can be more focused and effective. Initializing with all attributes does not remove the implausible attributes for the unpivot attribute set, leading to an unnecessarily large search space for subsequent refinement. Conversely, initializing with a random set often keeps implausible attributes and omits essential ones, which requires the refinement to simultaneously infer missing structure and correct mistakes, and thus significantly increases uncertainty and leads to unsatisfying results. For LLMs with fewer parameters, however, this advantage is not evident, and on some datasets, an initial set by querying LLM may lead to a worse result. This is because the weaker semantic capability of smaller models may generate an initial set with more errors. Considering the weaker Self-Refine capability, these errors may persist and lead to an unsatisfactory result. This limitation is most evident on the Gene dataset, where the scale of the dataset is large; in such cases, the initial set generated by small-scale models is often only marginally better, or even worse, than random initialization, resulting in substantially degraded accuracy. Moreover, on Gene, initializing with all attributes leads to an excessively long prompt that makes the $30$B model unable to return a valid completion, so we cannot obtain meaningful results and thus report an accuracy of $0$. In particular, the Adult dataset shows extremely low $Acc_\text{per\_attr.}$ with an initial set of all attributes. This is because its ground-truth unpivot attribute set is highly sparse ($2$ out of $19$ attributes), making the starting point far from the correct set. Faced with such a noisy and confusing initial set, the $30$B model struggles to refine effectively and tends to treat most attributes as unpivotable, resulting in notably lower accuracy.

\subsection{Effect of Similarity Metrics}
\label{app:similarity_metrics}

In this experiment, we evaluate the effect of different similarity metrics. We implement four methods to calculate the similarity score between two attributes, namely average, LLM score only ($\text{LLM}_\text{so}$), LLM weight only ($\text{LLM}_\text{wo}$) and LLM score and weight ($\text{LLM}_\text{sw}$). Specifically, the average method calculates the similarity score by applying a simple average to the similarity scores from different dimensions; $\text{LLM}_\text{so}$ calculates the score by directly asks the LLM for a final score based on table information; $\text{LLM}_\text{wo}$ provides the scores to LLM and asks it to generate a final score by giving these scores different weight accordingly; $\text{LLM}_\text{sw}$ requires LLM to first generate a similarity score between the two attributes according to its own semantic understanding of the tables, and then weight all scores to obtain the final result.

\begin{figure}[t]
    \begin{minipage}[h]{\linewidth}
        \centering
        \includegraphics[width=0.5\linewidth]{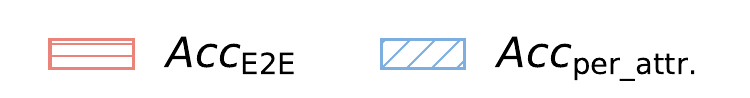}
    \end{minipage} \\
    \begin{minipage}[h]{0.49\linewidth}
        \centering
        \includegraphics[width=\linewidth]{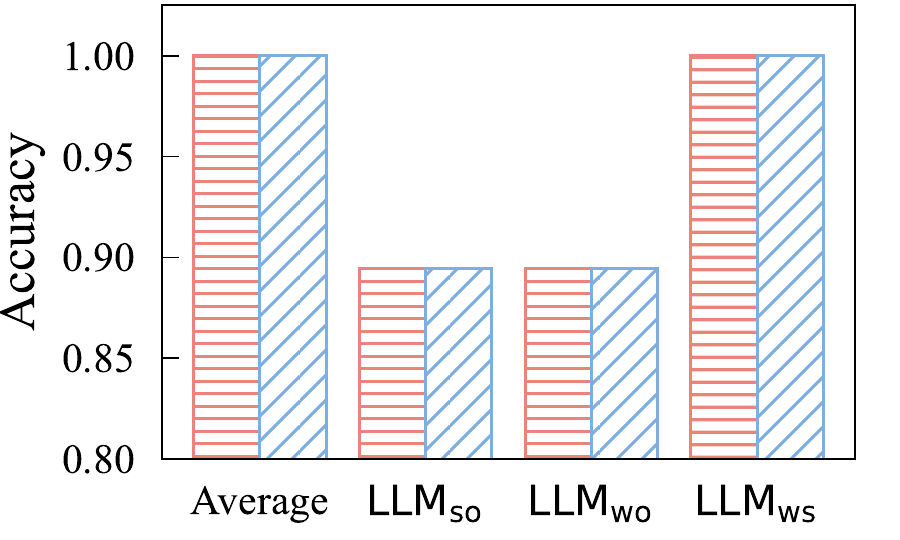}
        \subcaption{Adult}
    \end{minipage}
    \begin{minipage}[h]{0.49\linewidth}
        \centering
        \includegraphics[width=\linewidth]{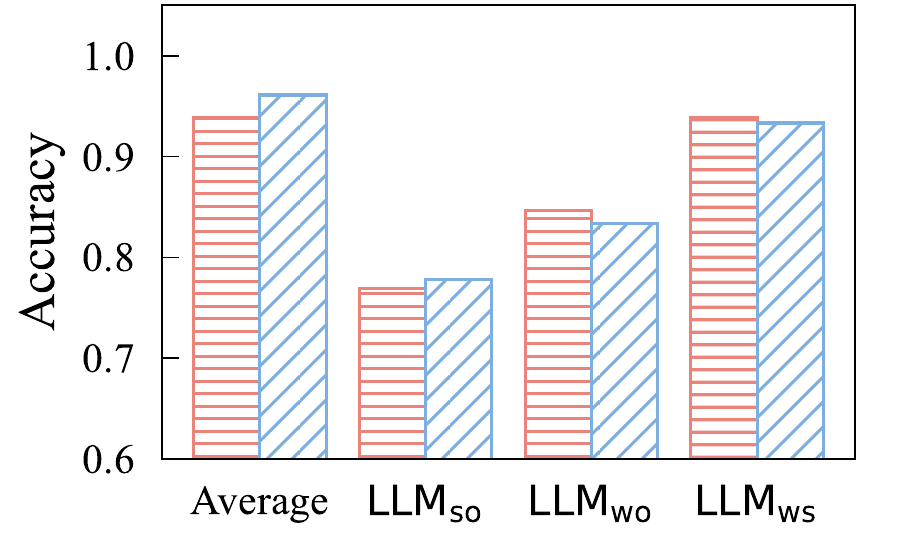}
        \subcaption{Football}
    \end{minipage} \\
    \begin{minipage}[h]{0.49\linewidth}
        \centering
        \includegraphics[width=\linewidth]{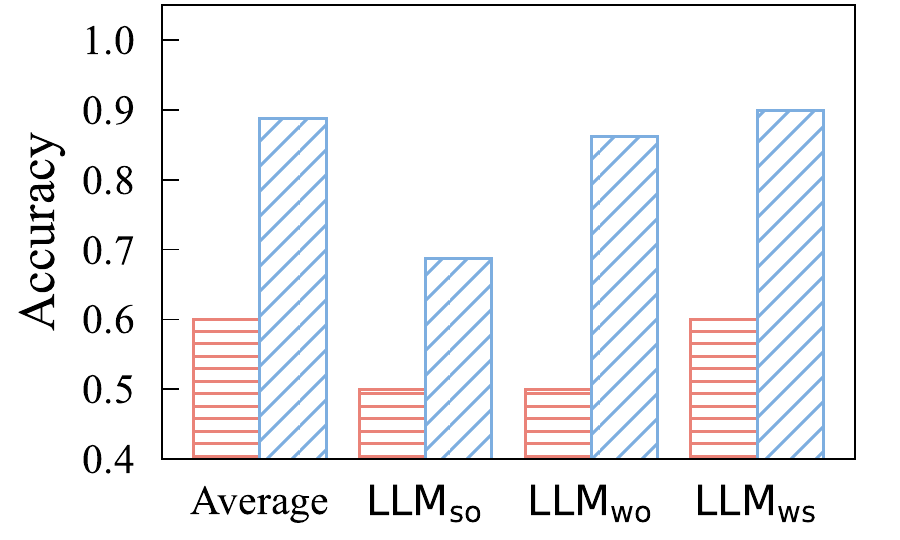}
        \subcaption{President}
    \end{minipage}
    \vspace{-3mm}
    \caption{Performance of different similarity metrics}
    \label{fig:similarity metrics}
    \vspace{-5mm}
\end{figure}

Figure~\ref{fig:similarity metrics} shows the accuracy of different similarity metrics on three datasets except the large-scale dataset Gene (issuing an LLM call for each attribute pair for large-scale datasets is unacceptable in terms of cost). Average performs the best on all the datasets, indicating that when it comes to numerical calculation and evaluation, LLM cannot provide a satisfying result. Rule-based similarity calculation is still a simple but useful method. Among the other three methods, $\text{LLM}_\text{so}$ performs the worst. This again demonstrates the limitation of LLM on similarity calculation tasks. $\text{LLM}_\text{wo}$ and $\text{LLM}_\text{sw}$ provide rule-based similarity scores to LLM, trying to make use of the semantic capability of LLM and the similarity capture capability of artificial rules. While this hybrid strategy is conceptually appealing, the experimental results suggest that LLM reasoning may interfere with consistent score aggregation rather than enhance it. Therefore, we choose the average method in \name~.

\subsection{Effect of LLM variants}

In this section we evaluate the effect of LLM variants. We run the experiment on 8 different LLMs, namely i) Qwen3-235B-A22B (Q3L), ii) Qwen3-32B (Q3M), iii) Qwen3-30B-A3B (Q3S), iv) DeepSeek-V3 (DV3), v) DeepSeek-R1 (DR1), vi) DeepSeek-R1-Distill-Llama-70b (DRL), vii) Qwen-Max-2025-01-25 (QM) and viii) OpenAI o3 (OO3). These models can be divided into 2 groups, the first 6 models are open-source models which we mainly compare considering the data privacy, and the last 2 models are closed-source models which we plot here as a reference and verify the effectiveness of \name~.

\begin{figure}[t]
    \begin{minipage}[h]{\linewidth}
        \centering
        \includegraphics[width=0.5\linewidth]{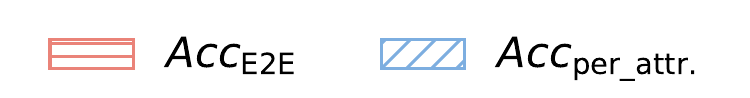}
    \end{minipage} \\
    \begin{minipage}[h]{0.49\linewidth}
        \centering
        \includegraphics[width=\linewidth]{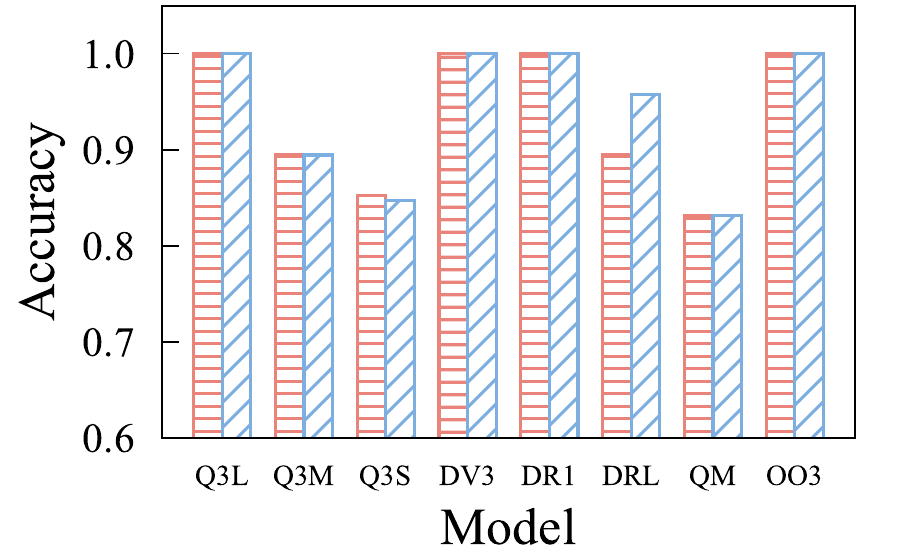}
        \subcaption{Adult}
    \end{minipage}
    \begin{minipage}[h]{0.49\linewidth}
        \centering
        \includegraphics[width=\linewidth]{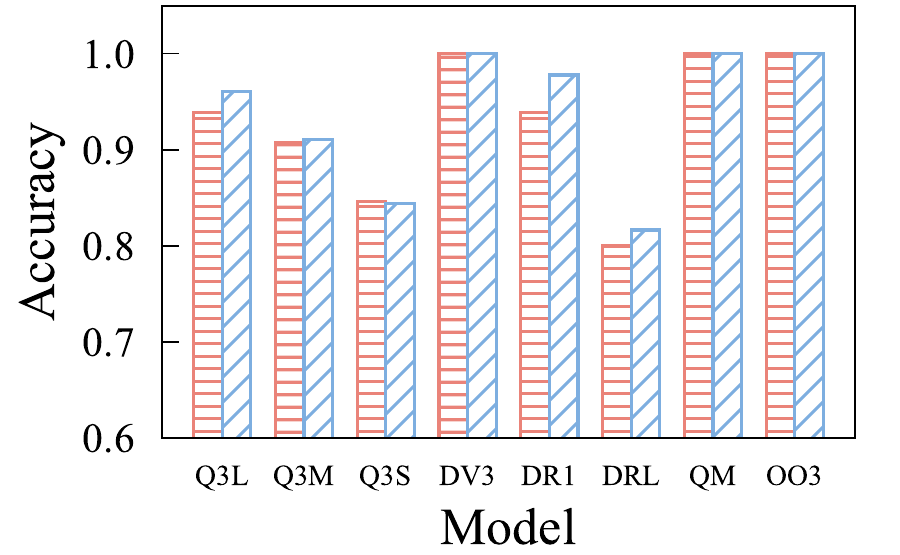}
        \subcaption{Football}
    \end{minipage} \\
    \begin{minipage}[h]{0.49\linewidth}
        \centering
        \includegraphics[width=\linewidth]{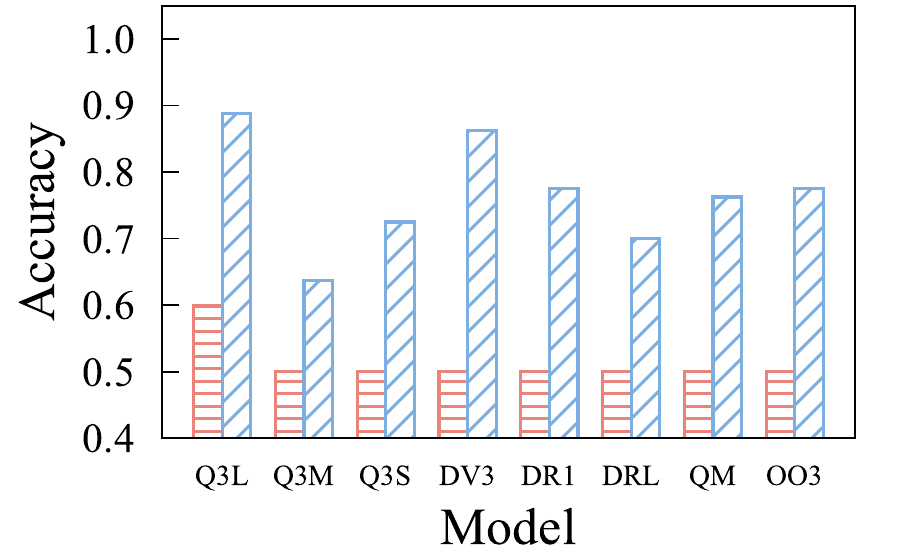}
        \subcaption{President}
    \end{minipage}
    \begin{minipage}[h]{0.49\linewidth}
        \centering
        \includegraphics[width=\linewidth]{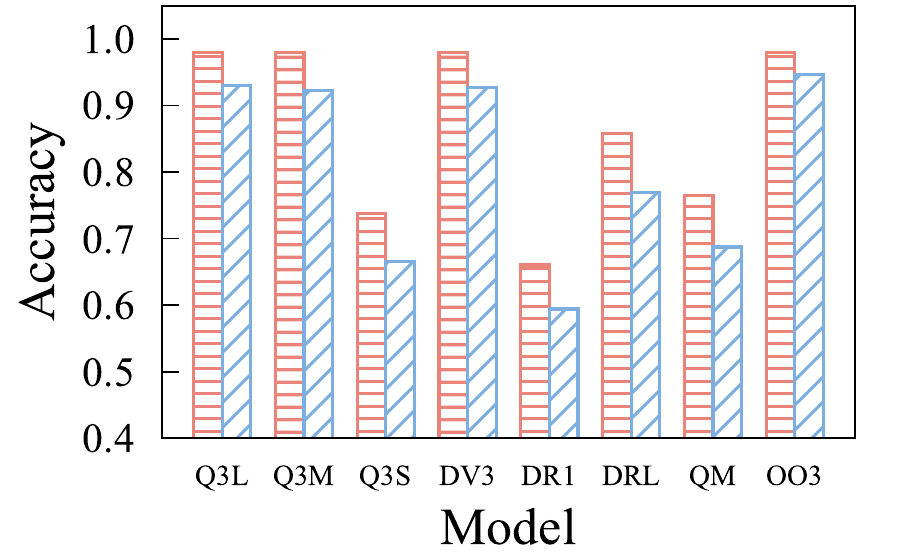}
        \subcaption{Gene}
    \end{minipage}
    \vspace{-3mm}
    \caption{Performance of different models}
    \label{fig:models}
    \vspace{-3mm}
\end{figure}

As shown in Figure~\ref{fig:models}, all the models can generate satisfying results. Although the accuracy drops as the model parameters reduce, the decline is minor and remains within an acceptable range. This allows users to freely select the model according to the usage scenario. In addition, the performance of open-source models is comparable to, and even surpasses, that of closed-source models on some datasets. This indicates that even when choosing open-source models under the constraints of data privacy requirements, the performance of \name~ will not deteriorate.









\end{document}